\documentclass[english,journal,12pt,draftclsnofoot,onecolumn]{IEEEtran}
\usepackage[T1]{fontenc}
\usepackage[latin9]{inputenc}
\usepackage{color}
\usepackage{float}
\usepackage{amsmath}
\usepackage{amsthm}
\usepackage{amssymb}
\usepackage{graphicx}
\usepackage{subfigure}

\makeatletter

\floatstyle{ruled}
\newfloat{algorithm}{tbp}{loa}
\providecommand{\algorithmname}{Algorithm}
\floatname{algorithm}{\protect\algorithmname}

\theoremstyle{plain}
\newtheorem{prop}{\protect\propositionname}
\theoremstyle{plain}
\newtheorem{cor}{\protect\corollaryname}
\theoremstyle{remark}
\newtheorem{rem}{\protect\remarkname}
\theoremstyle{plain}
\newtheorem{lem}{\protect\lemmaname}
\theoremstyle{plain}
\newtheorem{thm}{\protect\theoremname}

\usepackage{amsfonts}
\usepackage{cite
}\usepackage{array}
\usepackage{algorithm}
\usepackage{algorithmic}
\usepackage{subfigure}
\usepackage{stfloats}
\usepackage{graphicx}

\makeatother

\usepackage{babel}

\makeatother

\usepackage{babel}
\providecommand{\corollaryname}{Corollary}
\providecommand{\lemmaname}{Lemma}
\providecommand{\propositionname}{Proposition}
\providecommand{\remarkname}{Remark}
\providecommand{\theoremname}{Theorem}

\begin{document}

\title{Power Efficient IRS-Assisted NOMA}

\author{Jianyue Zhu, Yongming Huang, Jiaheng Wang, Keivan Navaie, Zhiguo
Ding\thanks{J. Zhu, Y. Huang, J. Wang are with the National Mobile Communications
Research Laboratory, Southeast University, Nanjing, China. (email:
\{zhujy, huangym, jhwang\}@seu.edu.cn).

K. Navaie is with the School of Computing and Communications, Lancaster
University, Lancaster, United Kingdom (email: k.navaie@lancaster.ac.uk). 

Z. Ding is with the School of Electrical and Electronic Engineering,
Manchester University, Manchester, UK (email: zhiguo.ding@manchester.ac.uk).}}
\maketitle
\begin{abstract}
In this paper, we propose a downlink multiple-input single-output
(MISO) transmission scheme, which is assisted by an intelligent reflecting
surface (IRS) consisting of a large number of passive reflecting elements.  In the literature, it has been
proved that nonorthogonal multiple access (NOMA) can achieve the capacity region when the channels
are quasi-degraded. However, in a conventional communication scenario, it is difficult to guarantee
the quasi-degradation, because the channels are determined
by the propagation environments and cannot be reconfigured. To overcome this difficulty, we
focus on an IRS-assisted MISO NOMA system, where the wireless channels can
be effectively tuned. We optimize the beamforming vectors and the IRS
phase shift matrix for minimizing  transmission power.  Furthermore, we propose
an improved quasi-degradation condition by using IRS, which can ensure that NOMA achieves the capacity region with
high possibility. For a comparison, we study zero-forcing beamforming (ZFBF) as well, where
the beamforming vectors and the IRS phase shift matrix are also jointly optimized.
Comparing NOMA with ZFBF, it is shown that, with the same IRS phase
shift matrix and the improved quasi-degradation condition, NOMA always
outperforms ZFBF. At the same time, we identify the condition under which ZFBF outperforms NOMA,  which motivates the proposed hybrid NOMA  transmission.   Simulation results show that the proposed IRS-assisted
MISO system outperforms the  MISO case without IRS,  and
the hybrid NOMA transmission scheme always achieves better performance than orthogonal
multiple access. 
\end{abstract}

\begin{IEEEkeywords}
Multiple-input single-output, nonorthogonal multiple access, intelligent
reflecting surface, zero-forcing beamforming, quasi-degradation
\end{IEEEkeywords}

\section{Introduction}

In the beyond fifth generation (B5G) communication systems, there are high
requirements in spectrum efficiency, energy consumption, and massive
connectivity \cite{saad2019vision,tariq2019speculative,8412482}.
In order to meet these high demands, various technologies, such as massive
multiple-input multiple-output (MIMO) \cite{guo2019convolutional},
millimeter wave \cite{xiao_millimeter_2017}, and small cell \cite{7835181}, 
are being investigated for the B5G communication systems. In addition,
nonorthogonal multiple access (NOMA) has also been introduced as a
promising multiple access candidate for future mobile networks \cite{dai2015non,surveynoma}.
Different from the conventional multiple access scheme, i.e., orthogonal
multiple access (OMA), NOMA allows multiple users sharing the same
resources, such as time, frequency, space, and code,  and hence significantly
improves the spectrum efficiency \cite{zhu_optimal_2017,7842433,ding2015cooperative}. 

In addition, being able to provide the flexibility and spatial degrees of
freedom needed for efficient resource allocation, multi-antenna techniques
have been widely considered in the B5G communication systems \cite{8654724,8352617}.
In the literature, there are many works focusing on the coexistence
of NOMA and multi-antenna techniques. For example, in \cite{ding2016application},
 the authors applied NOMA in MIMO systems and showed that MIMO-NOMA outperforms conventional MIMO-OMA. In \cite{z._chen_application_2016},
the authors proposed the quasi-degradation condition, under which NOMA can achieve the same performance as dirty paper coding
(DPC). This work also demonstrated that NOMA is not always preferable.
If users' channels are orthogonal, zero-forcing beamforming (ZFBF)
is a more preferred option than NOMA. On the other hand, the users' channels having the
same directions is the ideal case for the implementation of NOMA.

Conventionally, in NOMA systems, the directions of users' channels
cannot be tuned. This is because the users' channels are determined
by propagation environments and hence are highly stochastic.  Motivated by this, in this paper, we
apply the intelligent reflecting surface (IRS) to NOMA systems to reconfigure the channels. 

 The IRS, consisting of a large number of passive elements, is introduced as a promising technology for the B5G communication systems \cite{tang2018wireless,tang2019wireless}. The reflecting elements of the IRS  can affect the electromagnetic behavior of the wireless
propagation channel. Usually, the IRS operates as a multi-antenna relay
\cite{wu2018intelligent,huang2018energy} but it is fundamentally
different from a conventional relay. Specifically, the IRS functions as a reconfigurable
scatterer, which requires no dedicated energy resource for decoding,
channel estimation, and transmission. One of the key challenges in adopting IRS-assisted communication systems
is to design the IRS phase shift matrix \cite{huang2018energy,wu2018intelligent,wu2019beamforming,xu2019resource}.
In this paper, we  focus on the design of  beamforming vectors and the IRS phase shift matrix for an IRS-assisted multiple-input single-output
(MISO) NOMA communication system.

To the best of our knowledge, there is a limited number of works,
e.g., \cite{fu2019intelligent,ding2019simple,yang2019intelligent,mu2019exploiting},
investigating the application of IRS in NOMA systems. In \cite{fu2019intelligent},
the authors optimized the beamforming vectors and the IRS phase shift
matrix for minimizing downlink transmission power in the IRS-empowered
NOMA system. \cite{ding2019simple} proposed a simple design of IRS-assisted
NOMA transmission, where the authors characterized the performance
of practical IRS NOMA transmission. In \cite{yang2019intelligent},
the authors focused on the rate optimization problem for an IRS-assisted
NOMA system, where the power and the IRS phase shift matrix were optimized.
\cite{mu2019exploiting} maximized the sum rate of all users by jointly
optimizing the active beamforming at the BS and the passive beamforming
at the IRS, where the ideal and non-ideal IRS assumptions are both
studied. 

Different from all the above mentioned previous works, in this paper,
the advantage of the employment of IRS in NOMA systems is demonstrated
and the best achievable performance achieved by using IRS-assisted
NOMA is obtained. We optimize the beamforming vectors and the IRS
phase shift matrix to minimize the transmission power. Particularly,
for the formulated optimization problem, we consider a quasi-degradation
constraint to guarantee that NOMA achieves the same performance as
DPC \cite{z._chen_application_2016,jianyue}, which can realize the
capacity region \cite{weingarten_capacity_2006}. Conventionally,
whether the quasi-degradation condition can be satisfied depends on
the predetermined channels, which means NOMA might not always achieve
the capacity region. In addition, DPC is difficult to be implemented
in practical communication systems, due to its prohibitively high
complexity \cite{weingarten_capacity_2006}. By employing the IRS-assisted
NOMA, we design a more practical transmission scheme, which achieves
the same performance as DPC. We further investigate optimizing beamforming
vectors and the IRS phase shift matrix for IRS-assisted zero-forcing
beamforming (ZFBF) scheme and show that it is preferred to IRS-assisted
NOMA in certain condition.

In this paper, we study the joint optimization of beamforming vectors
and the IRS element matrix respectively for IRS-assisted NOMA and IRS-assisted
ZFBF and focus on the fundamental two users' case. The contributions
are summarized as follows. 
\begin{itemize}
\item In IRS-assisted NOMA systems, we study the optimization of beamforming
vectors and the IRS phase shift matrix for minimizing the BS transmission
power. Particularly, the quasi-degradation is considered as a constraint
in the optimization problem, which guarantees that NOMA can achieve
the same performance as DPC.
\item  For IRS-assisted NOMA, we provide an improved quasi-degradation condition
under which the feasibility of the optimization problem is guaranteed and the same
performance as DPC is obtained. 
\item We show that the improved quasi-degradation condition for IRS-assisted
NOMA can be satisfied with a higher possibility than conventional
NOMA without IRS.
\item The optimal beamforming solutions with given IRS phase shift matrix
are provided according to  our previous papers \cite{jianyue} and \cite{z._chen_application_2016}.
Then, with the optimal beamforming, we further optimize the IRS phase
shift matrix using semidefinite relaxation (SDR) and quadratic transform.
\item For IRS-assisted ZFBF, we also characterize the optimal beamforming
in a closed form with a given IRS phase shift matrix. The IRS phase
shift matrix is further optimized using fractional programing and
successive convex approximation (SCA).
\item Comparing IRS-assisted NOMA with IRS-assisted ZFBF, we show that,
given the same IRS phase shift matrix and the improved quasi-degradation
condition, NOMA achieves better performance than that of  ZFBF. However, ZFBF outperforms NOMA under certain condition and hence we provide the hybrid NOMA transmission scheme. 
\end{itemize}

The rest of the paper is organized as follows. In Section \ref{sec:System-Model},
we describe the system model and formulate the minimizing transmission
power problem. In Section \ref{sec:Beamforming-and-IRS}, we 
analyze the feasibility of the formulated problem in IRS-assisted NOMA systems,  and then the beamforming vectors and the IRS phase shift
matrix are optimized. In Section \ref{sec:Beamforming-and-IRS-1},
we design the beamforming vectors and the IRS phase shift matrix for
IRS-assisted ZFBF. In Section \ref{sec:The-Comparison-between}, we compare the
schemes of IRS-assisted NOMA and IRS-assisted ZFBF and propose the hybrid NOMA transmission scheme. Finally,
in Section \ref{sec:Simulation-Results}, we present simulation results followed by conclusion in Section \ref{sec:Conclusion}.

\textbf{\textcolor{black}{$\mathbf{\mathit{Notations}}$}}: The matrices
and vectors are respectively denoted as boldface capital and lower
case letters. Tr($\boldsymbol{A}$) and Rank($\boldsymbol{A}$) represent
the trace and rank of matrix $\boldsymbol{A}$, respectively; $\boldsymbol{a}^{H}$
is the Hermitian transpose vector $\boldsymbol{a}$; $\boldsymbol{A}\succeq\boldsymbol{0}$
indicates that $\boldsymbol{A}$ is a positive semidefinite matrix;
$\boldsymbol{I}_{N}$ is the $N\times N$ identity matrix; $\mathbb{C}^{T}$
denotes the set of complex numbers; $\left\Vert .\right\Vert $ denotes
the absolute value of the Euclidean vector norm.

\section{\label{sec:System-Model} System Model and Problem Formulation}

\subsection{System Model}

As shown in Fig. 1, we consider the downlink MISO system wherein a
base station (BS) equipped with $M$ antennas serves two single-antenna
users. In this MISO system, the communication is assisted by an IRS with $N$ phase shift elements.
Let $s_{i}$ be the message intended to be received by user $i$ with
$E\left[\left|s_{i}\right|^{2}\right]=1$ and $\boldsymbol{w}_{i}\in\mathbb{C}^{M}$
be the complex beamforming vector for user $i$. In this IRS-assisted
MISO system, we consider that the BS exploits superposition coding
and hence the received signal at each user is 
\begin{equation}
y_{k}=\boldsymbol{h}_{k}^{H}\left(\boldsymbol{w}_{1}s_{1}+\boldsymbol{w}_{2}s_{2}\right)+n_{k},k=1,2,
\end{equation}
where $n_{k}\sim\mathcal{CN}\left(0,\sigma^{2}\right)$ is the additive zero average 
white Gaussian noise (AWGN) with variance of $\sigma^{2}$ , and
\begin{equation}
\boldsymbol{h}_{k}^{H}=\mathbf{h}_{rk}^{H}\boldsymbol{\Theta}\mathbf{G}+\mathbf{h}_{dk}^{H}\in\mathbb{C}^{1\times M},\label{channel of user k}
\end{equation}
is the channel of user $k$. In \eqref{channel of user k}, $\mathbf{h}_{rk}^{H}\in\mathbb{C}^{1\times N}$
is the channel between the IRS and user $k$, $\mathbf{G}\in\mathbb{C}^{N\times M}$
denotes the channel between the BS and the IRS, and $\mathbf{h}_{dk}\in\mathbb{C}^{1\times M}$
represents the direct channel between the BS and user $k$; $\boldsymbol{\Theta}=\textrm{diag}\left(e^{j\theta_{1}},\cdots,e^{j\theta_{N}}\right)$
is the IRS phase shift matrix, where $j$ denotes the imaginary unit
and $\theta_{n}\in[0,2\pi]$ is the phase shift.

\begin{figure}
\centering
\includegraphics[scale=0.7]{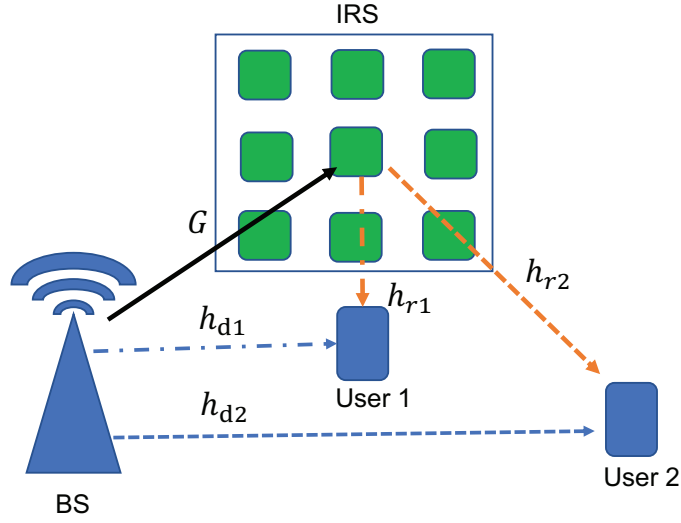}\caption{A downlink IRS-assisted MISO system.}
\end{figure}

\subsection{Problem Formulation}

\subsubsection{NOMA Transmission Scheme}

 It has been shown that with quasi-degraded
channels, the same performance as DPC, i.e., the capacity region of
MISO broadcast channels \cite{weingarten_capacity_2006}, can be achieved
by using NOMA \cite{chen_optimal_2016,z._chen_application_2016,jianyue}.
Considering the prohibitively high complexity, DPC is difficult
to be implemented in practical communication systems. Thus, we focus
on designing a more practical downlink transmission scheme, i.e.,
NOMA, which can outperform OMA and simultaneously yield a performance
close to the capacity region. 

It has been shown in \cite {z._chen_application_2016,jianyue,chen_optimal_2016} that whether the quasi-degradation condition can be
satisfied dependents on the channels between the BS and users, which
cannot be controlled. However, in this paper, by using
the IRS, the channels of users can be efficiently tuned and we might
find the IRS phase shift matrix, i.e., $\boldsymbol{\Theta}$, to
guarantee the quasi-degradation. Thus, considering the quasi-degradation
constraint with variable $\boldsymbol{\Theta}$ , the same performance
as DPC is achieved by using NOMA. 

In NOMA systems, the successive interference cancellation (SIC) is
employed for the users to decode their signals. In this paper, we
assume a fixed decoding order, (2, 1). Specifically, user 1 first decodes
the signal of user 2 and subtracts this from its received signal,
which means user 1 can decode its signal without interference. For
user 2, it does not perform SIC and simply decodes its signal by treating
user 1's signal as interference. Therefore, the achievable rate of
these two users are respectively given by
\begin{equation}
\begin{cases}
R_{1}=\ln\left(1+\frac{\boldsymbol{h}_{1}^{H}\boldsymbol{w}_{1}\boldsymbol{w}_{1}^{H}\boldsymbol{h}_{1}}{\sigma^{2}}\right),\\
R_{2}=\min\left\{ \ln\left(1+\textrm{SINR}_{2,1}\right),\ln\left(1+\textrm{SINR}_{2,2}\right)\right\} ,
\end{cases}
\end{equation}
where
\begin{equation}
\textrm{SINR}_{2,1}=\frac{\boldsymbol{h}_{1}^{H}\boldsymbol{w}_{2}\boldsymbol{w}_{2}^{H}\boldsymbol{h}_{1}}{\boldsymbol{h}_{1}^{H}\boldsymbol{w}_{1}\boldsymbol{w}_{1}^{H}\boldsymbol{h}_{1}+\sigma^{2}},
\end{equation}
\begin{equation}
\textrm{SINR}_{2,2}=\frac{\boldsymbol{h}_{2}^{H}\boldsymbol{w}_{2}\boldsymbol{w}_{2}^{H}\boldsymbol{h}_{2}}{\boldsymbol{h}_{2}^{H}\boldsymbol{w}_{1}\boldsymbol{w}_{1}^{H}\boldsymbol{h}_{2}+\sigma^{2}},
\end{equation}
respectively denote the signal interference noise ratio (SINR) of
user $1$ to decode user $2$ and user $2$ to decode itself. 

The performance of the IRS-assisted MISO NOMA system relies on the
design of beamforming vectors and the IRS phase shift matrix. In this
paper, we focus on the joint optimization of beamforming vectors and the
IRS phase matrix in  IRS-assisted MISO NOMA systems. We further 
 assume that the perfect channel state information (CSI) is available
at all nodes. In the formulated optimization problem, we consider
minimizing the transmission power (MTP) subject to the quality of service
(QoS) constraints of the two users. In addition,  we incorporate an additional
important constraint, i.e., the quasi-degradation constraint. This constraint
ensures that the NOMA scheme can achieve the same performance as DPC.
 The optimization problem is thus formulated as:
 \begin{subequations}
\begin{align}
 & P^{NOMA}=\underset{\left\{ \boldsymbol{w}_{1},\boldsymbol{w}_{2},\boldsymbol{\Theta}\right\} }{\min}\left\Vert \boldsymbol{w}_{1}\right\Vert ^{2}+\left\Vert \boldsymbol{w}_{2}\right\Vert ^{2},\label{P1-NOMA}\\
\textrm{s.t.} & ~\frac{\boldsymbol{h}_{1}^{H}\boldsymbol{w}_{1}\boldsymbol{w}_{1}^{H}\boldsymbol{h}_{1}}{\sigma^{2}}\geq r_{1}^{\min},\label{user_1_qos}\\
 & ~\min\left\{ \frac{\boldsymbol{h}_{1}^{H}\boldsymbol{w}_{2}\boldsymbol{w}_{2}^{H}\boldsymbol{h}_{1}}{\boldsymbol{h}_{1}^{H}\boldsymbol{w}_{1}\boldsymbol{w}_{1}^{H}\boldsymbol{h}_{1}+\sigma^{2}},\frac{\boldsymbol{h}_{2}^{H}\boldsymbol{w}_{2}\boldsymbol{w}_{2}^{H}\boldsymbol{h}_{2}}{\boldsymbol{h}_{2}^{H}\boldsymbol{w}_{1}\boldsymbol{w}_{1}^{H}\boldsymbol{h}_{2}+\sigma^{2}}\right\} \geq r_{2}^{\min},\label{user 2_qos}\\
 & ~\frac{1+r_{1}^{\min}}{\cos^{2}\alpha}-\frac{r_{1}^{\min}\cos^{2}\alpha}{\left(1+r_{2}^{\min}\left(1-\cos^{2}\alpha\right)\right)^{2}}\leq\frac{\left\Vert \boldsymbol{h}_{1}\right\Vert ^{2}}{\left\Vert \boldsymbol{h}_{2}\right\Vert ^{2}},\label{quasi-degradation}\\
 & ~0\leq\theta_{i}\leq2\pi,i=1,\cdots,N,\label{theta}
\end{align}
  \end{subequations}
where, $r_{i}^{\min}$, is the corresponding target signal-to-noise
ratio (SNR) level, i.e., $\ln\left(1+r_{i}^{\min}\right)=R_{i}^{\min}$,
$i=1,2$, and $\alpha$ is the angel between the two users' channels
with
\begin{equation}
\cos^{2}\alpha=\frac{\boldsymbol{h}_{1}^{H}\boldsymbol{h}_{2}\boldsymbol{h}_{2}^{H}\boldsymbol{h}_{1}}{\left\Vert \boldsymbol{h}_{1}\right\Vert ^{2}\left\Vert \boldsymbol{h}_{2}\right\Vert ^{2}}.
\end{equation}
In problem \eqref{P1-NOMA}, \eqref{user_1_qos} and \eqref{user 2_qos}
are respectively the QoS constraints of user 1 and user 2 and \eqref{theta}
is due to the IRS phase shift matrix, i.e., $\boldsymbol{\Theta}=\textrm{diag}\left(e^{j\theta_{1}},\cdots,e^{j\theta_{N}}\right)$.
In addition, \eqref{quasi-degradation} is the quasi-degradation constraint,
which was proposed in \cite{z._chen_application_2016,jianyue}. 

To the best of  our knowledge, there are very limited works
investigating the IRS-assisted NOMA system \cite{yang2019intelligent,ding2019simple,fu2019intelligent}.
The similar work in \cite{fu2019intelligent} also studied the optimization
of beamforming vectors and the IRS phase shift matrix for MTP. However,
\cite{fu2019intelligent} considered the optimization problem without
the quasi-degradation constraint, which means the proposed scheme might not achieve
the best performance using NOMA. In contrast here, we fully exploit the
advantage of IRS, i.e., the directions of users\textquoteright{} channel
can be effectively tuned to obtain the same performance as DPC.

\subsubsection{ZFBF Transmission Scheme}
Here, we benchmark our proposed IRS-assisted NOMA system against an IRS-assisted ZFBF system. ZFBF is a widely used transmission scheme in the literature. In the ZFBF strategy, the users transmit data in the null space of other users\textquoteright{}
channels, which mitigates the multi-user interference. If
the channels of the two users are orthogonal, the user interference
can be avoided by using ZFBF. Thus, we can achieve the best performance
using ZFBF if the channels of users are orthogonal. However, the
channels are not always orthogonal. In this paper, we investigate
the IRS-assisted ZFBF transmission scheme, where the IRS phase shift
matrix can be obtained to make the channels
orthogonal under certain condition.

In the IRS-assisted MISO ZFBF system, the MTP problem with QoS
constraints is formulated as:
\begin {subequations}
\begin{align}
 & P^{ZF}=\underset{\left\{ \boldsymbol{w}_{1},\boldsymbol{w}_{2},\boldsymbol{\Theta}\right\} }{\min}\left\Vert \boldsymbol{w}_{1}\right\Vert ^{2}+\left\Vert \boldsymbol{w}_{2}\right\Vert ^{2},\label{P1-ZF}\\
\textrm{s.t.}~ & \frac{\boldsymbol{h}_{1}^{H}\boldsymbol{w}_{1}\boldsymbol{w}_{1}^{H}\boldsymbol{h}_{1}}{\sigma^{2}}\geq r_{1}^{\min},\label{user 1_qos_zf}\\
 & \frac{\boldsymbol{h}_{2}^{H}\boldsymbol{w}_{2}\boldsymbol{w}_{2}^{H}\boldsymbol{h}_{2}}{\sigma^{2}}\geq r_{2}^{\min},\label{user 2_qos_zf}\\
 & \boldsymbol{h}_{1}^{H}\boldsymbol{w}_{2}=0,~\boldsymbol{h}_{2}^{H}\boldsymbol{w}_{1}=0,\label{interference avoidance}\\
 & 0\leq\theta_{i}\leq2\pi,i=1,\cdots,N,
\end{align}
\end {subequations}
where \eqref{user 1_qos_zf} and \eqref{user 2_qos_zf} are respectively
the QoS constraints of user 1 and user 2. In addition, the constraint
\eqref{interference avoidance} is due to the ZFBF principle, i.e.,
the users transmit data in the null space of other users\textquoteright{}
channels.

In previous works on IRS, e.g., \cite{ding2019simple,ZFBF-power},
also used the ZFBF beamforming.  \cite{ding2019simple} employed the IRS to make ZFBF an ideal beamforming for maximizing weak user's rate.  Also using ZFBF,  \cite{ZFBF-power} studied the joint optimization of power and the IRS phase shifting matrix for maximizing energy efficiency. In this
paper, we focus on the beamforming and the IRS phase shift matrix design to minimize the transmission power. Moreover, in the following, we will
show that, under certain condition, the IRS can be used to make the channels
orthogonal, which is the ideal case for ZFBF. 

\section{\label{sec:Beamforming-and-IRS}Beamforming and IRS Phase Shift 
Design: IRS-Assisted NOMA}

In this section, we first investigate the feasibility of the formulated
problem \eqref{P1-NOMA}. It is worthy pointing out that only if problem
\eqref{P1-NOMA} is feasible, the NOMA transmission can lead to the
same performance as DPC. Then, given the proposed feasibility condition,
we focus on the joint optimization of beamforming vectors and the IRS
phase shift matrix. Specifically, we first optimize the beamforming
vectors with given IRS phase shift matrix, $\boldsymbol{\Theta}$, and
then  $\boldsymbol{\Theta}$ is optimized. 

\subsection{\label{subsec:The-feasibility-of}The Feasibility of the Quasi-degradation
Constraint}

Different from the MISO NOMA system without IRS, the IRS-assisted
MISO NOMA system is able to adjust the angle between the channels
of user 1 and user 2 to satisfy the quasi-degradation condition. In
the following, we show that, compared to the literature MISO NOMA
system without IRS, the IRS-assisted MISO NOMA system has a greater possibility to
satisfy the quasi-degradation condition. Here, we derive a sufficient
condition to guarantee the feasibility of the quasi-degradation constraint.
\begin{prop}
\label{P:the feasibility-los}The quasi-degradation constraint \eqref{quasi-degradation}
is feasible if 
\begin{equation}
\lambda_{\max}\left(\boldsymbol{\Upsilon}_{1}-\boldsymbol{\Upsilon}_{2}\right)\geq0\label{sufficiency_quasi-degradation-los},
\end{equation}
where 
\begin{equation}
\boldsymbol{\Upsilon}_{k}=\left[\begin{array}{cc}
\boldsymbol{\Phi}_{k}\boldsymbol{\Phi}_{k}^{H} & \boldsymbol{\Phi}_{k}\mathbf{h}_{dk}\\
\mathbf{h}_{dk}^{H}\boldsymbol{\Phi}_{k}^{H} & \mathbf{h}_{dk}^{H}\mathbf{h}_{dk}
\end{array}\right],\label{gama_k}
\end{equation}
 $\boldsymbol{\Phi}_{k}=\textrm{diag}\left(\mathbf{h}_{rk}\right)\mathbf{G}$
for $k=1,2$, and $\lambda_{\max}\left(\boldsymbol{\Upsilon}_{1}-\boldsymbol{\Upsilon}_{2}\right)$
denotes the maximum eigenvalue of matrix $\boldsymbol{\Upsilon}_{1}-\boldsymbol{\Upsilon}_{2}$. 
\end{prop}
\begin{proof}
See Appendix A. 
\end{proof}
According to Proposition \ref{P:the feasibility-los}, given the improved
quasi-degradation condition \eqref{sufficiency_quasi-degradation-los},
the IRS phase shift matrix can always be found to satisfy the quasi-degradation
constraint \eqref{quasi-degradation}. Thus, under condition \eqref{sufficiency_quasi-degradation-los},
the proposed MTP problem in \eqref{P1-NOMA} using NOMA is always feasible
and the IRS-assisted MISO NOMA scheme can obtain the same performance
as DPC. 
\begin{cor}
\label{C-q}Suppose $\boldsymbol{h}_{1}^{H}\boldsymbol{h}_{2}\neq0$,
the quasi-degradation condition, \eqref{quasi-degradation}, is always satisfied
if the IRS is located at the position of user 1.
\end{cor}
\begin{proof}
See Appendix B.
\end{proof}
In practice, the location of the IRS can be easily controlled \cite{di2019smart}.
Hence, to realize the capacity region using NOMA, the IRS would better
be located close to user 1. 
\begin{rem}
In the following, we further show that, adopting IRS, the quasi-degradation
condition would be satisfied with a greater possibility. 

To better illustrate the advantage of the employment of IRS in MISO
NOMA systems, we provide the following example. In Fig. 2 and Fig.
3, the BS is located at point $\left(0,0\right)$, the IRS central
element is located at point $\left(5,5\right)$ and the single-antenna
users are randomly placed within a rectangular around the BS. The
channel between the IRS and user $k$ is given by $\mathbf{h}_{rk}=d_{rk}^{-\alpha}\mathbf{g}_{rk}$,
where $d_{rk}$ is the distance between user $k$ and the BS, $\alpha=3$
is the path loss exponent, $\mathbf{g}_{rk}$ follows a Rayleigh distribution.
Similarly, the channel between the BS and the IRS is $\mathbf{G}=d_{r}^{-\alpha}\mathbf{g}_{r}$,
and the direct channel between the BS and user $k$ is $\mathbf{h}_{dk}=d_{k}^{-\alpha}\mathbf{g}_{k}$.
Assume that the location of user 1 is fixed, we focus on describing
different locations of user 2 to satisfy the quasi-degradation condition
without IRS, i.e., 
\begin{equation}
\frac{1+r_{1}^{\min}}{\cos^{2}\alpha}-\frac{r_{1}^{\min}\cos^{2}\alpha}{\left(1+r_{2}^{\min}\left(1-\cos^{2}\alpha\right)\right)^{2}}\leq\frac{\left\Vert \mathbf{h}_{d1}\right\Vert ^{2}}{\left\Vert \mathbf{h}_{d2}\right\Vert ^{2}},
\end{equation}
where 
\begin{equation}
\cos^{2}\alpha=\frac{\mathbf{h}_{d1}^{H}\mathbf{h}_{d2}\mathbf{h}_{d2}^{H}\mathbf{h}_{d1}}{\left\Vert \mathbf{h}_{d1}\right\Vert ^{2}\left\Vert \mathbf{h}_{d2}\right\Vert ^{2}},
\end{equation}
and the improved quasi-degradation condition provided in
Proposition \ref{P:the feasibility-los} as \eqref{sufficiency_quasi-degradation-los}.
In addition, in Fig. 2 and Fig. 3, the black area denotes the locations
of user 2 satisfying the quasi-degradation condition without IRS or
the improved quasi-degradation condition.

Fig. 2 and Fig. 3 (a)  respectively show the region of user 2 to satisfy
quasi-degradation condition without IRS, and the improved quasi-degradation
condition. In these two figures, the location of user 1 is fixed
at $\left(5,5.5\right)$. From Fig. 2 and Fig.3 (a), one can easily find
that, with the help of IRS, the black area becomes larger. Therefore,
compared with the conventional MISO NOMA systems without IRS, the
quasi-degradation condition is more likely to be satisfied in IRS-assisted
MISO NOMA systems. Thus, the use of IRS in MISO systems facilitates to achieve the best performance by using
NOMA transmission scheme. 

Moreover, Fig. 3 (b)  depicts the region of user 2 to satisfy the
improved quasi-degradation condition. Different from Fig. 3 (a), the location
of user 1 is fixed at $\left(5,5.2\right)$. Hence, compared with
Fig. 3 (a), the location of user 1 in Fig. 3 (b) is closer to the IRS. From
Fig. 3 (a)  and Fig. 3 (b), it is easy to see that if user 1 is near the IRS,
the quasi-degradation condition is more likely to be satisfied. Actually,
according to Corollary \ref{C-q}, if user 1 and the IRS are at the same
location, the quasi-degradation condition can always be satisfied.
\begin{figure}
\centering
\includegraphics[scale=0.5]{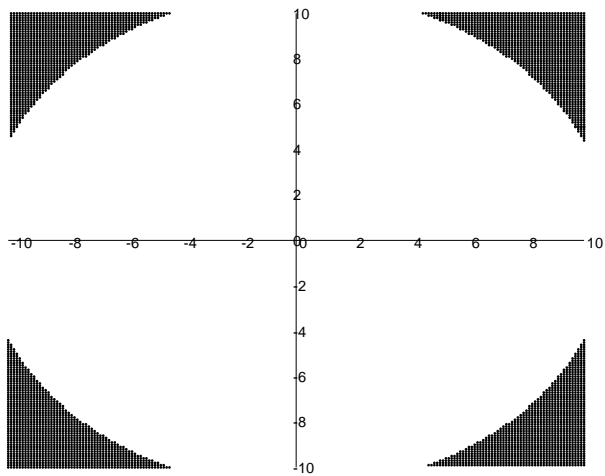}
\caption{Quasi-degradation region of user 2 without using IRS with fixed location of
user 1 at $(5,5.5)$.}
\end{figure}
\begin{figure}
\centering 
\subfigure[Fixed locations of user 1 at (5, 5.5) and IRS at (5, 5).]{
\includegraphics[scale=0.4]{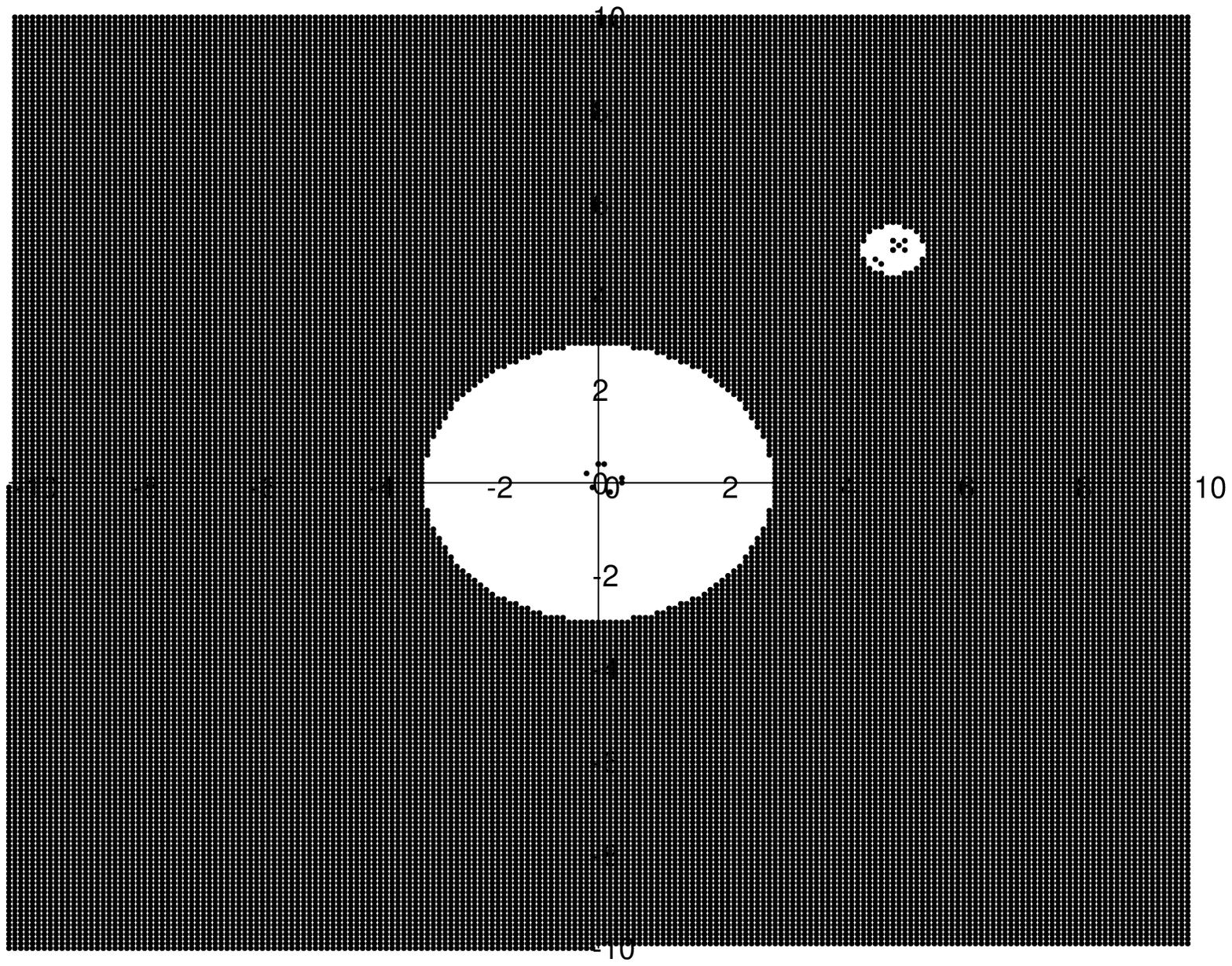}}
\subfigure[Fixed locations of user 1 at (5, 5.2) and IRS at (5, 5).]{
\includegraphics[scale=0.4]{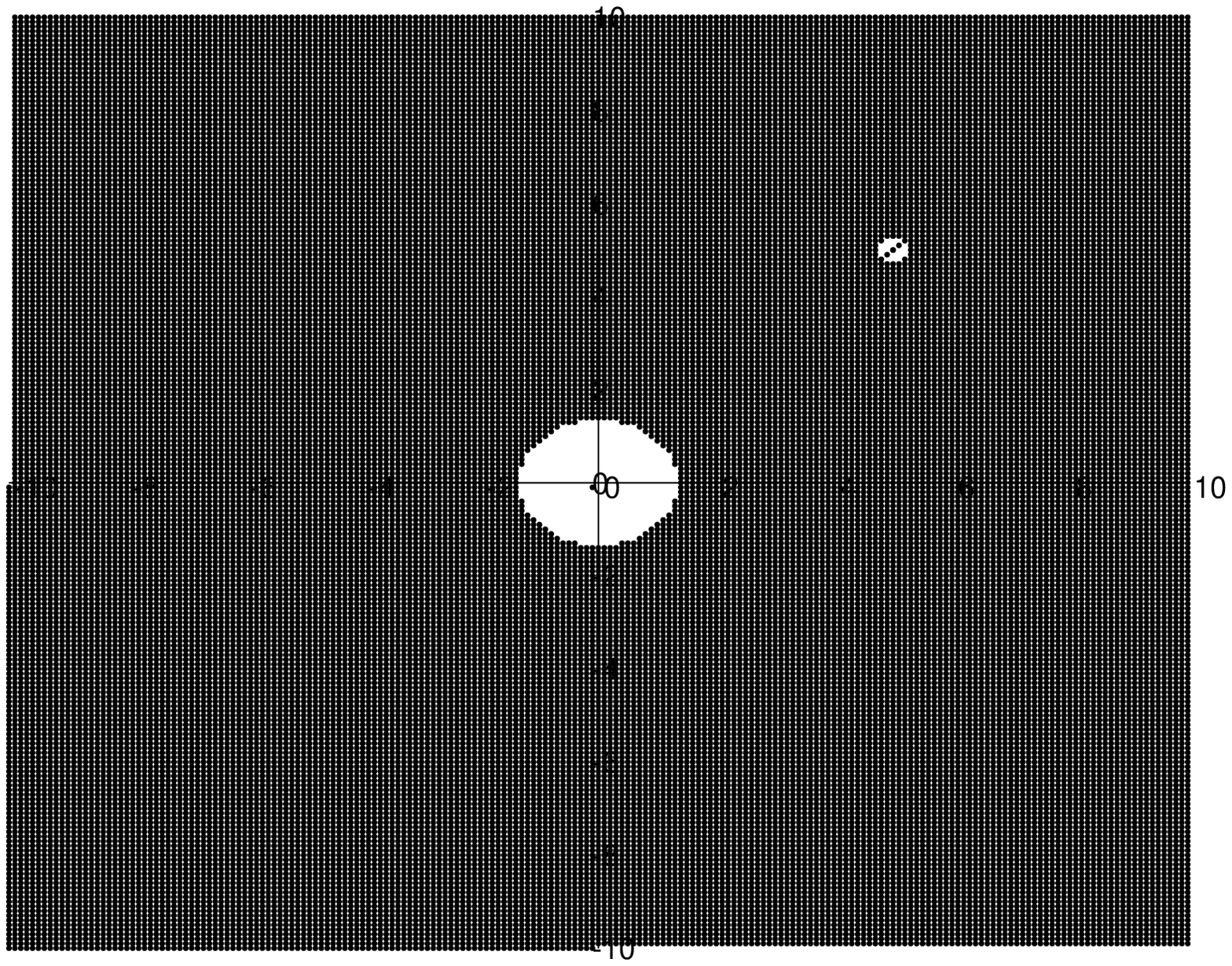}}
\caption{Improved quasi-degradation region of user 2.}
\end{figure}
\end{rem}

\subsection{ Beamforming Vectors and the IRS Phase Shift Matrix Design for MTP}

In this subsection, we focus on optimizing the beamforming vectors
and the IRS phase shift matrix for MTP. Given \eqref{sufficiency_quasi-degradation-los},
the proposed optimization problem \eqref{P1-NOMA} is always feasible.
However, problem \eqref{P1-NOMA} is a non-convex problem due to the
non-convex constraints with respect to $\boldsymbol{w}_{1}$, $\boldsymbol{w}_{2}$,
and $\boldsymbol{\Theta}$. In the following, we first optimize the
beamforming vectors, i.e., $\boldsymbol{w}_{1}$, $\boldsymbol{w}_{2}$,
for any given $\boldsymbol{\Theta}$ and the results are given in
the following Lemma.
\begin{lem}
\label{L-optimal beamforming}Given the improved quasi-degradation condition \eqref{sufficiency_quasi-degradation-los},
for any given $\boldsymbol{\Theta}$, the optimal beamforming solutions
to problem \eqref{P1-NOMA} is 
\begin{equation}
\begin{cases}
\boldsymbol{w}_{1}^{*}=\phi_{1}\left(1+A_{2}\right)\boldsymbol{e}_{1}-r_{2}^{\min}\boldsymbol{e}_{2}^{H}\boldsymbol{e}_{1}\boldsymbol{e}_{2},\\
\boldsymbol{w}_{2}^{*}=\phi_{2}\boldsymbol{e}_{2},
\end{cases}\label{optimal beamforming}
\end{equation}
where
\begin{equation}
\begin{cases}
\boldsymbol{e}_{1}=\frac{\boldsymbol{h}_{1}}{\left\Vert \boldsymbol{h}_{1}\right\Vert },\boldsymbol{e}_{2}=\frac{\boldsymbol{h}_{2}}{\left\Vert \boldsymbol{h}_{2}\right\Vert },\\
\phi_{1}^{2}=\frac{r_{1}^{\min}\sigma^{2}}{\left\Vert \boldsymbol{h}_{1}\right\Vert ^{2}}\frac{1}{\left(1+r_{2}^{\min}\sin^{2}\alpha\right)^{2}},\\
\phi_{2}^{2}=\frac{r_{2}^{\min}\sigma^{2}}{\left\Vert \boldsymbol{h}_{2}\right\Vert ^{2}}+\frac{r_{1}^{\min}\sigma^{2}}{\left\Vert \boldsymbol{h}_{1}\right\Vert ^{2}}\frac{r_{2}^{\min}\cos^{2}\alpha}{\left(1+r_{2}^{\min}\sin^{2}\alpha\right)^{2}}.
\end{cases}.\label{fi_1,fi_2}
\end{equation}
\end{lem}
\begin{proof}
See the proof of Proposition 1 in \cite{z._chen_application_2016}.
\end{proof}
From Lemma \ref{L-optimal beamforming}, for any given $\boldsymbol{\Theta}$,
the optimal beamforming is characterized in a closed form expression. Then, using
the closed-form beamforming solutions, we optimize the IRS
phase shift matrix, i.e., $\boldsymbol{\Theta}=\textrm{diag}\left(e^{j\theta_{1}},\cdots,e^{j\theta_{N}}\right)$.
The corresponding optimization problem is
\begin {subequations}
\begin{align}
\underset{\left\{ \boldsymbol{\Theta}\right\} }{\min}~ & F,\label{P7}\\
\textrm{s.t.}~ & \frac{1+r_{1}^{\min}}{\cos^{2}\alpha}-\frac{r_{1}^{\min}\cos^{2}\alpha}{\left(1+r_{2}^{\min}\left(1-\cos^{2}\alpha\right)\right)^{2}}\leq\frac{\left\Vert \boldsymbol{h}_{1}\right\Vert ^{2}}{\left\Vert \boldsymbol{h}_{2}\right\Vert ^{2}},\label{quasi_degradation_P7-1}\\
 & 0\leq\theta_{i}\leq2\pi,i=1,\cdots,N,
\end{align}
\end {subequations}
where
\begin{align}
F & =\left\Vert \boldsymbol{w}_{1}^{*}\right\Vert ^{2}+\left\Vert \boldsymbol{w}_{2}^{*}\right\Vert ^{2},\nonumber \\
 & =\phi_{1}^{2}\left(\left(1+r_{2}^{\min}\right)^{2}-\left(2+r_{2}^{\min}\right)r_{2}^{\min}\cos^{2}\alpha\right)+\phi_{2}^{2},\label{objective function-1}
\end{align}
and $\phi_{1}^{2},~\phi_{2}^{2}$ have been written in \eqref{fi_1,fi_2}. 

To solve problem \eqref{P7}, we introduce variable $\boldsymbol{Q}=\tilde{\boldsymbol{v}}\tilde{\boldsymbol{v}}^{H}$, 
with $\widetilde{\boldsymbol{v}}=\left[\boldsymbol{v};1\right]$, and
$\boldsymbol{v}=[e^{j\theta_{1}},\cdots,e^{j\theta_{N}}]^{H}$. Therefore, for $k=1, 2$,  $\boldsymbol{h}_{k}^{H}=\mathbf{h}_{rk}^{H}\boldsymbol{\Theta}\mathbf{G}+\mathbf{h}_{dk}^{H}=\boldsymbol{v}^{H}\boldsymbol{\Phi}_{k}+\mathbf{h}_{dk}^{H}$, with $\boldsymbol{\Phi}_{k}=\textrm{diag}\left(\mathbf{h}_{rk}\right)\mathbf{G}$
and 
\begin{equation}
\left\Vert \boldsymbol{h}_{k}\right\Vert ^{2}=\textrm{ Tr}\left(\boldsymbol{Q}\boldsymbol{\Upsilon}_{k}\right),\label{h_k'*h_k-1-1}
\end{equation}
\begin{equation}
\boldsymbol{h}_{1}^{H}\boldsymbol{h}_{2}=\textrm{Tr}\left(\boldsymbol{Q}\boldsymbol{R}\right),\label{h1'*h2-2}
\end{equation}
where $\boldsymbol{\Upsilon}_{k}$ has been given in \eqref{gama_k}
and 
\begin{equation}
\boldsymbol{R}=\left[\begin{array}{cc}
\boldsymbol{\Phi}_{1}\boldsymbol{\Phi}_{2}^{H} &\boldsymbol{\Phi}_{1}\mathbf{h}_{d2} \\
\mathbf{h}_{d1}^{H}\boldsymbol{\Phi}_{2}^{H}& \mathbf{h}_{d1}^{H}\mathbf{h}_{d2}
\end{array}\right].\label{R-1}
\end{equation}
Therefore, problem \eqref{P7} is reduced to the following
\begin {subequations}
\begin{align}
\underset{\left\{ \boldsymbol{Q}\succeq\boldsymbol{0}\right\} }{\min}~ & \phi_{1}^{2}\left(\left(1+r_{2}^{\min}\right)^{2}-\left(2+r_{2}^{\min}\right)r_{2}^{\min}\cos^{2}\alpha\right)+\phi_{2}^{2},\label{P7-1}\\
\textrm{s.t.}~ & \frac{1+r_{1}^{\min}}{\cos^{2}\alpha}-\frac{r_{1}^{\min}\cos^{2}\alpha}{\left(1+r_{2}^{\min}\left(1-\cos^{2}\alpha\right)\right)^{2}}\leq\frac{\textrm{Tr}\left(\boldsymbol{Q}\boldsymbol{\Upsilon}_{1}\right)}{\textrm{Tr}\left(\boldsymbol{Q}\boldsymbol{\Upsilon}_{2}\right)},\label{quasi_degradation_P7-1-1}\\
 & \boldsymbol{Q}_{i,i}=1,i=1,\cdots,N+1,\label{Q_i,i}\\
 & \textrm{Rank}\left(\boldsymbol{Q}\right)=1,\label{rank(Q)}
\end{align}
\end{subequations}
where
\begin{equation}
\cos^{2}\alpha=\frac{\textrm{Tr}\left(\boldsymbol{Q}\boldsymbol{R}\right)\textrm{Tr}\left(\boldsymbol{R}^{H}\boldsymbol{Q}^{H}\right)}{\textrm{ Tr}\left(\boldsymbol{Q}\boldsymbol{\Upsilon}_{1}\right)\textrm{ Tr}\left(\boldsymbol{Q}\boldsymbol{\Upsilon}_{2}\right)},
\end{equation}
\begin{equation}
\phi_{1}^{2}=\frac{r_{1}^{\min}\sigma^{2}}{\textrm{ Tr}\left(\boldsymbol{Q}\boldsymbol{\Upsilon}_{1}\right)}\frac{1}{\left(1+r_{2}^{\min}\sin^{2}\alpha\right)^{2}},
\end{equation}
and 

\begin{equation}
\phi_{2}^{2}=\frac{r_{2}^{\min}\sigma^{2}}{\textrm{ Tr}\left(\boldsymbol{Q}\boldsymbol{\Upsilon}_{2}\right)}+\frac{r_{1}^{\min}\sigma^{2}}{\textrm{ Tr}\left(\boldsymbol{Q}\boldsymbol{\Upsilon}_{1}\right)}\frac{r_{2}^{\min}\cos^{2}\alpha}{\left(1+r_{2}^{\min}\sin^{2}\alpha\right)^{2}}.
\end{equation}
Note that the constraints in \eqref{Q_i,i} and \eqref{rank(Q)} come
from the introduced variable substitution: $\boldsymbol{Q}=\tilde{\boldsymbol{v}}\tilde{\boldsymbol{v}}^{H}$.
Problem \eqref{P7-1} is very difficult to solve due to the tricky
objective function and the quasi-degradation constraint. In practice, the objective function and the quasi-degradation constraint are both nonconvex due to the complex expression of $\cos^{2}\alpha$.  In the following,
we first deal with the quasi-degradation constraint, which will be
transformed into a convex constraint.
\begin{prop}
\label{P_quasi-degradation}The quasi-degradation constraint can be
satisfied if 
\begin{equation}
\boldsymbol{\chi}_{1}\boldsymbol{Q}\boldsymbol{\chi}_{1}^{H}-\left(1+r_{1}^{\min}\right)\boldsymbol{\chi}_{2}\boldsymbol{Q}\boldsymbol{\chi}_{2}^{H}\succeq\boldsymbol{0},\label{quasi-transfer}
\end{equation}
where $\boldsymbol{\chi}_{k}=\left[\mathbf{\Phi}_{k}^{H},\mathbf{h}_{dk}\right]$,
$k=1,2$.
\end{prop}
\begin{proof}
See Appendix C.
\end{proof}
The quasi-degradation constraint \eqref{quasi-transfer} in Proposition \ref{P_quasi-degradation} is convex. We then focus on the complex nonconvex objective function. An
upper bound of the objective function is given in Proposition \ref{P_upper bound}. 
\begin{prop}
\label{P_upper bound}The upper bound of the objective function in
problem \eqref{P7} is given by 
\begin{equation}
F\leq\frac{\sigma^{2}\left(r_{1}^{\min}\left(1+r_{2}^{\min}\right)\right)}{\textrm{ Tr}\left(\boldsymbol{Q}\boldsymbol{\Upsilon}_{1}\right)}+\frac{\sigma^{2}r_{2}^{\min}}{\textrm{ Tr}\left(\boldsymbol{Q}\boldsymbol{\Upsilon}_{2}\right)}.
\end{equation}
\end{prop}
\begin{proof}
See Appendix D.
\end{proof}
Using Proposition \ref{P_quasi-degradation} and Proposition \ref{P_upper bound},
an approximation solution to problem \eqref{P7-1} is obtained by solving
the following problem
\begin {subequations}
\begin{align}
\underset{\left\{ \boldsymbol{Q}\succeq\boldsymbol{0}\right\} }{\min}~ & \frac{\sigma^{2}\left(r_{1}^{\min}\left(1+r_{2}^{\min}\right)\right)}{\textrm{ Tr}\left(\boldsymbol{Q}\boldsymbol{\Upsilon}_{1}\right)}+\frac{\sigma^{2}r_{2}^{\min}}{\textrm{ Tr}\left(\boldsymbol{Q}\boldsymbol{\Upsilon}_{2}\right)},\label{P4}\\
\textrm{s.t.}~ & \boldsymbol{\chi}_{1}\boldsymbol{Q}\boldsymbol{\chi}_{1}^{H}-\left(1+r_{1}^{\min}\right)\boldsymbol{\chi}_{2}\boldsymbol{Q}\boldsymbol{\chi}_{2}^{H}\succeq\boldsymbol{0}, \\
 & \boldsymbol{Q}_{i,i}=1,i=1,\cdots,N+1, \\
 & \textrm{Rank}\left(\boldsymbol{Q}\right)=1.
\end{align}
\end{subequations}
Note that \eqref{P4} is actually a sum of ratios problem. In order
to solve this multiple-ratio fractional programming, we exploit the
quadratic transform, which is described in the following Lemma \cite{8314727}.
\begin{lem}
\label{quadratic transform}\cite{8314727}The sum-of-ratio
problem 
\begin{align}
\underset{\left\{ \boldsymbol{x}\right\} }{\max}~ & \sum_{m=1}^{M}\frac{A_{m}\left(\boldsymbol{x}\right)}{B_{m}\left(\boldsymbol{x}\right)},\\
\textrm{s.t.} & ~\boldsymbol{x}\in\mathcal{X},\nonumber 
\end{align}
is equivalent to 
\begin{align}
\underset{\left\{ \boldsymbol{x},\boldsymbol{y}\right\} }{\max}~ & \sum_{m=1}^{M}\left(2y_{m}\sqrt{A_{m}\left(\boldsymbol{x}\right)}-y_{m}^{2}B_{m}\left(\boldsymbol{x}\right)\right),\\
\textrm{s.t.}~ & \boldsymbol{x}\in\mathcal{X},y_{m}\in\mathbb{R},m=1,\cdots,M,\nonumber 
\end{align}
where $\boldsymbol{y}$ refers to a collection of variables $\left\{ y_{1},\cdots,y_{M}\right\} $. 
\end{lem}
Using Lemma \ref{quadratic transform}, problem \eqref{P4} is then reduced to
\begin {subequations}
\begin{align}
\underset{\left\{ \boldsymbol{Q}\succeq\boldsymbol{0},y_{1},y_{2}\right\} }{\max} & f\label{P5},\\
\textrm{s.t.} & ~\boldsymbol{\chi}_{1}\boldsymbol{Q}\boldsymbol{\chi}_{1}^{H}-\left(1+r_{1}^{\min}\right)\boldsymbol{\chi}_{2}\boldsymbol{Q}\boldsymbol{\chi}_{2}^{H}\succeq\boldsymbol{0},\label{quasi-P5}\\
 & \boldsymbol{Q}_{i,i}=1,i=1,\cdots,N+1,\\
 & \textrm{Rank}\left(\boldsymbol{Q}\right)=1,\label{rank-1}
\end{align}
\end {subequations}
where 
\[
f=-2y_{1}\sqrt{\sigma^{2}\left(r_{1}^{\min}\left(1+r_{2}^{\min}\right)\right)}+y_{1}^{2}\textrm{ Tr}\left(\boldsymbol{Q}\boldsymbol{\Upsilon}_{1}\right)-2y_{2}\sqrt{\sigma^{2}r_{2}^{\min}}+y_{2}^{2}\textrm{ Tr}\left(\boldsymbol{Q}\boldsymbol{\Upsilon}_{2}\right).
\]
Here, we propose to optimize the primal variable $\boldsymbol{Q}$
and the auxiliary variables $y_{1},$ $y_{2}$ iteratively. Specifically,
we first optimize $\boldsymbol{Q}$ by fixing the auxiliary variables
$y_{1},$ $y_{2}$ and then further optimize $y_{1},$ $y_{2}$. However,
by fixing $y_{1},$ $y_{2}$, \eqref{P5} is also nonconvex
due to the rank constraint \eqref{rank-1}. Here, we exploit the semidefinite relaxation (SDR) \cite{ma2010semidefinite}
method to solve the nonconvex problem \eqref{P5} and the relaxed
problem is therefore: 
\begin {subequations}
\begin{align}
\underset{\left\{ \boldsymbol{Q}\succeq\boldsymbol{0}\right\} }{\max}~ & y_{1}^{2}\textrm{Tr}\left(\boldsymbol{Q}\boldsymbol{\Upsilon}_{1}\right)+y_{2}^{2}\textrm{Tr}\left(\boldsymbol{Q}\boldsymbol{\Upsilon}_{2}\right),\label{P8}\\
\textrm{s.t.}~ & \boldsymbol{\chi}_{1}\boldsymbol{Q}\boldsymbol{\chi}_{1}^{H}-\left(1+r_{1}^{\min}\right)\boldsymbol{\chi}_{2}\boldsymbol{Q}\boldsymbol{\chi}_{2}^{H}\succeq\boldsymbol{0},\\
 & \boldsymbol{Q}_{i,i}=1,i=1,\cdots,N+1.
\end{align}
\end{subequations}

One can easily show that problem \eqref{P8} is a standard semidefinite
programming (SDP), which can be efficiently solved through convex
optimization tools, e.g., CVX. However, the solution to problem \eqref{P8}
cannot always satisfy the rank constraint, i.e., $\textrm{Rank}\left(\boldsymbol{Q}\right)=1$.
In general, if the solution to the relaxed problem in \eqref{P8}
is a set of rank-one matrices, then it will be also the optimal solution
to the problem \eqref{P5}. Otherwise, the randomization technique, see, e.g., 
\cite{huang_rank-constrained_2010}, can be used to generate a set
of rank-one solutions. 

In addition, with fixed $\boldsymbol{Q}$, the optimal $y_{1},$
$y_{2}$  are obtained in closed form as
\begin{equation}
y_{1}^{*}=\frac{\sqrt{\sigma^{2}\left(r_{1}^{\min}\left(1+r_{2}^{\min}\right)\right)}}{\textrm{Tr}\left(\boldsymbol{Q}\boldsymbol{\Upsilon}_{1}\right)},\label{y_1_*}
\end{equation}
\begin{equation}
y_{2}^{*}=\frac{\sqrt{\sigma^{2}r_{2}^{\min}}}{\textrm{Tr}\left(\boldsymbol{Q}\boldsymbol{\Upsilon}_{2}\right)}.\label{y_2_*}
\end{equation}

Therefore, the optimal $\boldsymbol{Q}$ can thus be efficiently obtained
through numerical convex optimization. The steps of the proposed approach are summarized
in Algorithm 1. Additional steps are needed to obtain the beamforming vectors and  IRS elements. Specifically, we first obtain the eigenvalue
decomposition of $\boldsymbol{Q}$ as $\boldsymbol{Q}=\boldsymbol{U}\boldsymbol{\Lambda}\boldsymbol{U}$,
where $\boldsymbol{U}=\left[e_{1},\cdots,e_{N+1}\right]$ is a unitary
matrix and $\boldsymbol{\Lambda}=\textrm{diag}\left(\lambda_{1},\cdots,\lambda_{N}\right)$
is a diagonal matrix. Therefore, the IRS elements are given as 
\begin{equation}
e^{j\theta_{k}}=e^{j\textrm{arg}\left(\frac{\tilde{\boldsymbol{v}}_{k}}{\tilde{\boldsymbol{v}}_{N+1}}\right)},k=1,\cdots,N. \label{IRS elements}
\end{equation}
In \eqref{IRS elements}, $\tilde{\boldsymbol{v}}=\boldsymbol{U}\boldsymbol{\Lambda}^{1/2}\boldsymbol{r}$
and $\boldsymbol{r}\in\mathbb{C}^{\left(N+1\right)\times1}$ is chosen from a lot of randomly generalized vectors satisfying $\boldsymbol{r}\in\mathcal{CN}\left(\boldsymbol{0},\boldsymbol{I}_{N+1}\right)$.
Then, the beamforming vectors can be easily obtained using Lemma \ref{L-optimal beamforming}.

\begin{algorithm}
\caption{Iterative Approach for solving \eqref{P5} }

1: Initialize $y_{1},$ $y_{2}$ to a feasible value.

2:\textbf{Repeat}

3: \hspace*{1cm}Update $\boldsymbol{Q}$ by solving the SDR problem in
\eqref{P8}. 

4: \hspace*{1cm}Update $y_{1}$, $y_{2}$ by using \eqref{y_1_*}
and \eqref{y_2_*}.

5:\textbf{Until} The value of objective function in \eqref{P8} converges or the maximum number of iterations is reached.
\end{algorithm}

\section{\label{sec:Beamforming-and-IRS-1}Beamforming and IRS Phase Shift Design:  IRS-Assisted ZFBF}

The transmission scheme of ZFBF is in fact a special case of OMA
as the spatial degrees of freedom are used for interference avoidance.
Here, we investigate the IRS-assisted ZFBF transmission
scheme, where the beamforming vectors and the IRS phase shift matrix are
optimized to achieve the minimum transmission power. 

In order to solve problem \eqref{P1-ZF}, we first optimize the beamforming
vectors,  $\boldsymbol{w}_{1}$ and $\boldsymbol{w}_{2}$, by
fixing $\boldsymbol{\Theta}$.  Lemma \ref{L-ZF-B} gives the optimal beamforming solutions.
\begin{lem}
\label{L-ZF-B}For any given $\boldsymbol{\Theta}$, the optimal beamforming
solutions to problem in \eqref{P1-ZF}  are
\begin{equation}
\begin{cases}
\boldsymbol{w}_{1}^{*}=\frac{\sqrt{r_{1}^{\min}}\boldsymbol{a}_{1}}{\left\Vert \boldsymbol{h}_{1}\right\Vert ^{2}\left\Vert \boldsymbol{h}_{2}\right\Vert ^{2}\sin^{2}\alpha},\\
\boldsymbol{w}_{2}^{*}=\frac{\sqrt{r_{2}^{\min}}\boldsymbol{a}_{2}}{\left\Vert \boldsymbol{h}_{1}\right\Vert ^{2}\left\Vert \boldsymbol{h}_{2}\right\Vert ^{2}\sin^{2}\alpha},
\end{cases}
\end{equation}
where 
\begin{equation}
\sin^{2}\alpha=1-\frac{\boldsymbol{h}_{1}^{H}\boldsymbol{h}_{2}\boldsymbol{h}_{2}^{H}\boldsymbol{h}_{1}}{\left\Vert \boldsymbol{h}_{1}\right\Vert ^{2}\left\Vert \boldsymbol{h}_{2}\right\Vert ^{2}},
\end{equation}
and
\begin{equation}
\begin{cases}
\boldsymbol{a}_{1}=\left\Vert \boldsymbol{h}_{2}\right\Vert ^{2}\boldsymbol{h}_{1}-\left(\boldsymbol{h}_{2}^{H}\boldsymbol{h}_{1}\right)\boldsymbol{h}_{2},\\
\boldsymbol{a}_{2}=-\left(\boldsymbol{h}_{1}^{H}\boldsymbol{h}_{2}\right)\boldsymbol{h}_{1}+\left\Vert \boldsymbol{h}_{1}\right\Vert ^{2}\boldsymbol{h}_{2}.
\end{cases}
\end{equation}
\end{lem}
\begin{proof}
See Appendix E.
\end{proof}
From Lemma \ref{L-ZF-B}, with a given $\boldsymbol{\Theta}$, the optimal
beamforming solutions to problem \eqref{P1-ZF} is characterized in
closed form expressions. Thereby, using the closed-form beamforming solutions,
we further optimize the IRS phase shift matrix and the corresponding
optimization problem is
\begin{subequations}
\begin{align}
\underset{\left\{ \boldsymbol{\Theta}\right\} }{\min}~ & W=\frac{\sigma^{2}}{\sin^{2}\alpha}\left(\frac{r_{1}^{\min}}{\left\Vert \boldsymbol{h}_{1}\right\Vert ^{2}}+\frac{r_{2}^{\min}}{\left\Vert \boldsymbol{h}_{2}\right\Vert ^{2}}\right),\label{P2-ZF}\\
\textrm{s.t.} & ~0\leq\theta_{i}\leq2\pi,i=1,\cdots,N.
\end{align}
\end{subequations}
Similarly, to solve problem \eqref{P2-ZF}, we also use the introduced
variables, i.e., $\boldsymbol{Q}=\tilde{\boldsymbol{v}}\tilde{\boldsymbol{v}}^{H}$
with $\widetilde{\boldsymbol{v}}=\left[\boldsymbol{v};1\right]$ and
$\boldsymbol{v}=[e^{j\theta_{1}},\cdots,e^{j\theta_{N}}]^{H}$. Then,
problem \eqref{P2-ZF} is rewritten as
\begin {subequations}
\begin{align}
\underset{\left\{ \boldsymbol{Q}\succeq\boldsymbol{0}\right\} }{\min}~ & W=\frac{\sigma^{2}}{\sin^{2}\alpha}\left(\frac{r_{1}^{\min}}{\textrm{Tr}\left(\boldsymbol{Q}\boldsymbol{\Upsilon}_{1}\right)}+\frac{r_{2}^{\min}}{\textrm{Tr}\left(\boldsymbol{Q}\boldsymbol{\Upsilon}_{1}\right)}\right),\label{P2-ZF-1}\\
\textrm{s.t.} & ~\boldsymbol{Q}_{i,i}=1,i=1,\cdots,N+1,\\
 & ~\textrm{Rank}\left(\boldsymbol{Q}\right)=1,
\end{align}
\end {subequations}
where 
\begin{equation}
\sin^{2}\alpha=1-\frac{\textrm{Tr}\left(\boldsymbol{Q}\boldsymbol{R}\right)\textrm{Tr}\left(\boldsymbol{R}^{H}\boldsymbol{Q}^{H}\right)}{\textrm{ Tr}\left(\boldsymbol{Q}\boldsymbol{\Upsilon}_{1}\right)\textrm{ Tr}\left(\boldsymbol{Q}\boldsymbol{\Upsilon}_{2}\right)}.
\end{equation}
The objective function in \eqref{P2-ZF-1} is a nonconvex and  fractional objective, which is rewritten as
\begin{align}
W & =\frac{\sigma^{2}}{\sin^{2}\alpha}\left(\frac{r_{1}^{\min}}{\textrm{Tr}\left(\boldsymbol{Q}\boldsymbol{\Upsilon}_{1}\right)}+\frac{r_{2}^{\min}}{\textrm{Tr}\left(\boldsymbol{Q}\boldsymbol{\Upsilon}_{2}\right)}\right)\nonumber \\
 & =\frac{\sigma^{2}r_{1}^{\min}\textrm{Tr}\left(\boldsymbol{Q}\boldsymbol{\Upsilon}_{2}\right)+\sigma^{2}r_{2}^{\min}\textrm{Tr}\left(\boldsymbol{Q}\boldsymbol{\Upsilon}_{1}\right)}{\textrm{Tr}\left(\boldsymbol{Q}\boldsymbol{\Upsilon}_{1}\right)\textrm{Tr}\left(\boldsymbol{Q}\boldsymbol{\Upsilon}_{2}\right)-\textrm{Tr}\left(\boldsymbol{Q}\boldsymbol{R}\right)\textrm{Tr}\left(\boldsymbol{R}^{H}\boldsymbol{Q}^{H}\right)}.\label{P_ZF}
\end{align}
 In order to solve the problem in \eqref{P2-ZF-1}, we introduce the following objective
function:
\begin{equation}
G\left(\boldsymbol{Q},\eta\right)=G_{1}\left(\boldsymbol{Q}\right)-\eta G_{2}\left(\boldsymbol{Q}\right),\label{G}
\end{equation}
where $\eta$ is a positive parameter, 
\begin{equation}
G_{1}\left(\boldsymbol{Q}\right)=\sigma^{2}r_{1}^{\min}\textrm{Tr}\left(\boldsymbol{Q}\boldsymbol{\Upsilon}_{2}\right)+\sigma^{2}r_{2}^{\min}\textrm{Tr}\left(\boldsymbol{Q}\boldsymbol{\Upsilon}_{1}\right),
\end{equation}
and 
\begin{equation}
G_{2}\left(\boldsymbol{Q}\right)=\textrm{Tr}\left(\boldsymbol{Q}\boldsymbol{\Upsilon}_{1}\right)\textrm{Tr}\left(\boldsymbol{Q}\boldsymbol{\Upsilon}_{2}\right)-\textrm{Tr}\left(\boldsymbol{Q}\boldsymbol{R}\right)\textrm{Tr}\left(\boldsymbol{R}^{H}\boldsymbol{Q}^{H}\right).
\end{equation}
 We then study the following optimization problem with given $\eta$:
 \begin {subequations}
\begin{align}
\underset{\left\{ \boldsymbol{Q}\succeq\boldsymbol{0}\right\} }{\min}~ & G\left(\boldsymbol{Q},\eta\right)=G_{1}\left(\boldsymbol{Q}\right)-\eta G_{2}\left(\boldsymbol{Q}\right),\label{P3-ZF}\\
\textrm{s.t.}~ & \boldsymbol{Q}_{i,i}=1,i=1,\cdots,N+1,\\
 & \textrm{Rank}\left(\boldsymbol{Q}\right)=1,
\end{align}
\end {subequations}
The relation between problem \eqref{P2-ZF-1} and \eqref{P3-ZF} is
given by the following lemma.
\begin{lem}
\label{L-DK}\cite{dinkelbach1967nonlinear}Let
$G^{*}\left(\eta\right)$ be the optimal objective value of problem
\eqref{P3-ZF} and $\boldsymbol{Q}^{*}\left(\eta\right)$ be the optimal
solution of problem \eqref{P3-ZF}. Then, $\boldsymbol{Q}^{*}\left(\eta\right)$
is the optimal solution to \eqref{P2-ZF-1} if and only if $G^{*}\left(\eta\right)=0$. 
\end{lem}
According to Lemma \ref{L-DK}, the optimal solution to \eqref{P2-ZF-1}
can be found by solving problem \eqref{P3-ZF} parameterized by $\eta$
and then updating $\eta$ until $G^{*}\left(\eta\right)=0$. Therefore, we first focus on solving \eqref{P3-ZF} for a given $\eta$
. 

For a  given $\eta$, the non-convexity of problem \eqref{P3-ZF} lies on  the rank constraint
and part of the  objective function $G_{2}\left(\boldsymbol{Q},\eta\right)$.
In the following, for a  given $\eta$, we will solve the nonconvex
    problem \eqref{P3-ZF} by using successive approximation (SCA) \cite{kurtaran2002crashworthiness} and SDR.

Firstly, we construct global under-estimators of $G_{2}\left(\boldsymbol{Q},\eta\right)$.
Specifically, for any feasible point $\boldsymbol{Q}^{i}$, the function
$G_{2}\left(\boldsymbol{Q},\eta\right)$ satisfies the following inequality:
\begin{align}
G_{2}\left(\boldsymbol{Q},\eta\right) & \geq G_{2}\left(\boldsymbol{Q}^{i}\right)+\textrm{Tr}\left(\left(\nabla_{\boldsymbol{Q}}G_{2}\left(\boldsymbol{Q}^{i}\right)\right)^{H}\left(\boldsymbol{Q}-\boldsymbol{Q}^{i}\right)\right)\nonumber \\
 & =\tilde{G_{2}}\left(\boldsymbol{Q},\boldsymbol{Q}^{i}\right),
\end{align}
where 
\begin{equation}
\nabla_{\boldsymbol{Q}}G_{2}\left(\boldsymbol{Q}^{i}\right)=\boldsymbol{\Upsilon}_{1}^{H}\textrm{Tr}\left(\boldsymbol{Q}^{i}\boldsymbol{\Upsilon}_{2}\right)+\boldsymbol{\Upsilon}_{2}^{H}\textrm{Tr}\left(\boldsymbol{Q}^{i}\boldsymbol{\Upsilon}_{1}\right)-2\boldsymbol{R}^{H}\textrm{Tr}\left(\boldsymbol{Q}^{i}\boldsymbol{R}\right).
\end{equation}
Therefore, for any given $\boldsymbol{Q}^{i}$, an upper bound of
problem \eqref{P3-ZF} can be obtained by solving the following optimization
problem 
\begin {subequations}
\begin{align}
\underset{\left\{ \boldsymbol{\boldsymbol{Q}\succeq0}\right\} }{\min}~ & G\left(\boldsymbol{Q},\eta\right)=G_{1}\left(\boldsymbol{Q}\right)-\eta\tilde{G_{2}}\left(\boldsymbol{Q},\boldsymbol{Q}^{i}\right),\label{P5-ZF}\\
\textrm{s.t.}~ & \boldsymbol{Q}_{i,i}=1,i=1,\cdots,N+1,\\
 & ~\textrm{Rank}\left(\boldsymbol{Q}\right)=1.\label{rank}
\end{align}
\end{subequations}
Optimization problem \eqref{P5-ZF} 's  non-convexity only lies on the rank constraint \eqref{rank}.
To tackle this issue, we remove the rank constraint \eqref{rank}
by applying SDR. Then, the relaxed version of \eqref{P5-ZF} can be
efficiently solved via standard optimization tools, such as CVX. 

Note that  \eqref{P5-ZF} serves as an upper bound of
\eqref{P3-ZF}. The process of SCA is summarized in Algorithm 2, where
we iteratively tighten the upper bound and obtain a sequence of solutions,
i.e., $\boldsymbol{Q}$. 

\begin{algorithm}
\caption{Successive Convex Approximation Algorithm}

1:\textbf{ Initialize} iteration index $i=1$ and a feasible $\boldsymbol{Q}^{1}$
.

2:\textbf{ Repeat}

3: \hspace*{1cm}Solve problem \eqref{P5-ZF} for a given $\boldsymbol{Q}^{i}$
and store the solution $\boldsymbol{Q}$;

4: \hspace*{1cm}Set $i=i+1$ and $\boldsymbol{Q}^{i}=\boldsymbol{Q}$;

5:\textbf{ Until} The value of objective function in \eqref{P5-ZF} converges or the maximum number of iterations is reached.

6: $\boldsymbol{Q}^{*}=\boldsymbol{Q}^{i}$. 
\end{algorithm}

After the solution to problem \eqref{P3-ZF} is obtained, we shall
find an $\eta$ such that $G^{*}\left(\eta\right)=0$, which can be
achieved by Algorithm 3.  In Algorithm 3, $\delta$ denotes a small positive threshold. The  algorithm is guaranteed to converge to the desirable
$\eta$. Finally, we can achieve the solution of IRS elements by using
the SVD of $\boldsymbol{Q}^{*}$ and then the beamforming vectors
are obtained using Lemma \ref{L-ZF-B}.

\begin{algorithm}
\caption{The solution to problem \eqref{P3-ZF}}

1:\textbf{ Initialize} $\eta_{ini}=0$, $G^{*}\left(\eta\right)=\infty$, 
and precision $\delta>0$.

2:\textbf{ While $\left|G^{*}\left(\eta\right)\right|>\delta$ do }

3: \hspace*{1cm}Find the solution $\boldsymbol{Q}^{*}$ using Algorithm
2;

4: \hspace*{1cm}Calculate $G^{*}\left(\eta\right)$; 

5:\textbf{ }\hspace*{1cm}Update $\eta=\frac{G_{1}\left(\boldsymbol{Q}^{*}\right)}{G_{2}\left(\boldsymbol{Q}^{*}\right)}$; 

6: \textbf{Return} $\eta$ and $\boldsymbol{Q}^{*}$. 
\end{algorithm}

\section{\label{sec:The-Comparison-between}Comparison of IRS-Assisted
NOMA and IRS-Assisted ZFBF}

From the previous sections, we have respectively investigated the
joint optimization of beamforming vectors and the IRS phase shift matrix
for IRS-assisted NOMA and IRS-assisted ZFBF.  The application
of IRS in multi-antenna systems facilitates the implementation of
NOMA and ZFBF, which is because the directions of users' channel vectors
can be effectively aligned. In this section, we further study the
comparison between schemes of IRS-assisted NOMA and IRS-assisted ZFBF.

According to Lemma \ref{L-optimal beamforming}, with improved quasi-degraded
condition, the minimum transmission power achieved by NOMA is 
\begin{align}
P^{NOMA} & =\phi_{1}^{2}\left(\left(1+r_{2}^{\min}\right)^{2}-\left(2+r_{2}^{\min}\right)r_{2}^{\min}\cos^{2}\alpha\right)+\phi_{2}^{2}\nonumber \\
 & =\frac{r_{1}^{\min}\left(1+r_{2}^{\min}\right)\sigma^{2}}{\left\Vert \boldsymbol{h}_{1}\right\Vert ^{2}\left(1+r_{2}^{\min}\sin^{2}\alpha\right)}+\frac{r_{2}^{\min}\sigma^{2}}{\left\Vert \boldsymbol{h}_{2}\right\Vert ^{2}}.
\end{align}
In addition, from Lemma \ref{L-ZF-B}, the minimum transmission power
achieved by ZFBF is given by 
\begin{equation}
P^{ZFBF}=\frac{\sigma^{2}}{\sin^{2}\alpha}\left(\frac{r_{1}^{\min}}{\left\Vert \boldsymbol{h}_{1}\right\Vert ^{2}}+\frac{r_{2}^{\min}}{\left\Vert \boldsymbol{h}_{2}\right\Vert ^{2}}\right).
\end{equation}
In the following, we compare $P^{NOMA}$ and $P^{ZFBF}$.
\begin{thm}
\label{P-p_noma<=00003Dp_zf}Given the same IRS phase shift matrix
$\boldsymbol{\Theta}$ and improved quasi-degraded condition, we always
have 
\begin{equation}
P^{NOMA}\leq P^{ZFBF}.
\end{equation}
\end{thm}
\begin{proof}
See Appendix F.

According to Proposition \ref{P-p_noma<=00003Dp_zf}, given the same
IRS phase shift matrix and the improved quasi-degraded condition,
the scheme of NOMA always achieves better performance than ZFBF. In
addition, in the following, we consider a specific situation where
the channels of user 1 and user 2 are orthogonal, i.e., $\boldsymbol{h}_{1}^{H}\boldsymbol{h}_{2}=0$.
\end{proof}
\begin{prop}
\label{P_ZF_condition}The improved quasi-degradation condition is
invalid when  $\boldsymbol{h}_{1}^{H}\boldsymbol{h}_{2}=0$
, which can be satisfied if and only if 
\begin{equation}
\boldsymbol{R}=\left[\begin{array}{cc}
\boldsymbol{\Phi}_{1}\boldsymbol{\Phi}_{2}^{H} & \mathbf{h}_{d1}^{H}\boldsymbol{\Phi}_{2}^{H}\\
\boldsymbol{\Phi}_{1}\mathbf{h}_{d2} & \mathbf{h}_{d1}^{H}\mathbf{h}_{d2}
\end{array}\right]=\boldsymbol{0},\label{Q=00003D0}
\end{equation}
 and ZFBF achieves the optimal value.
\end{prop}
\begin{proof}
See Appendix G.
\end{proof}
From Proposition \ref{P_ZF_condition}, the IRS phase shift matrix
can be found to make the channels of user 1 and user 2 orthogonal
only with condition \eqref{Q=00003D0}. In practice, the condition \eqref{Q=00003D0}
is very difficult to satisfy due to the random $\mathbf{h}_{rk}$,
$\mathbf{h}_{dk}$, for $k=1,2$, and $\mathbf{G}$. However, as shown
in Subsection \ref{subsec:The-feasibility-of}, the improved quasi-degradation
condition for MISO NOMA is more easily satisfied, especially
for cases where the IRS is located close to user 1.
\begin{rem}
Based on the provided results, the hybrid NOMA precoding scheme can
be given as follows:
\end{rem}
\begin{itemize}
\item Given the improved quasi-degradation condition, which  can be satisfied
with high possibility, NOMA achieves the optimal performance and hence
we prefer to use NOMA transmission scheme.
\item If $\boldsymbol{R}=\boldsymbol{0}$, under which the IRS phase shift
matrix can be found to satisfy $\boldsymbol{h}_{1}^{H}\boldsymbol{h}_{2}=0$,
the ZFBF is preferred. 
\item Assume a case where the improved quasi-degradation condition is violated
and $\boldsymbol{R}\neq\boldsymbol{0}$, the performance loss is inevitable.
Considering the computational complexity, one may prefer to use ZFBF as the closed form beamforming cannot be obtained by using
NOMA. 
\end{itemize}

\section{\label{sec:Simulation-Results}Simulation Results}

In this section, the performance of the proposed solution to the MTP
problem is evaluated. In our simulation, the BS is located at the
point $\left(0\textrm{m},0\textrm{m}\right)$, and single-antenna
users are randomly placed in the half right-hand side square $\left(200\textrm{m}\times200\textrm{m}\right)$
around the BS. The channel between the IRS and user $k$ is given
by $\mathbf{h}_{rk}=d_{rk}^{-\alpha}\mathbf{g}_{rk}$, where $d_{rk}$
is the distance between user $k$ and the BS, $\alpha=3$ is the path
loss exponent, $\mathbf{g}_{rk}$ follows a Rayleigh distribution.
Similarly, the channel between the BS and the IRS is $\mathbf{G}=d_{r}^{-\alpha}\mathbf{g}_{r}$,
and the direct channel between the BS and user $k$ is $\mathbf{h}_{dk}=d_{k}^{-\alpha}\mathbf{g}_{k}$. The noise
power is given by $\sigma^{2}=BN_{0}$ with bandwidth $B=10\textrm{MHz}$
and the noise power spectral density $N_{0}=-174$dBm. The target
rates for both users are the same ($R_{k}^{\min}=1\textrm{bps/Hz}$
for $k=1,2$). 

\begin{figure}
\centering
\includegraphics[scale=0.7]{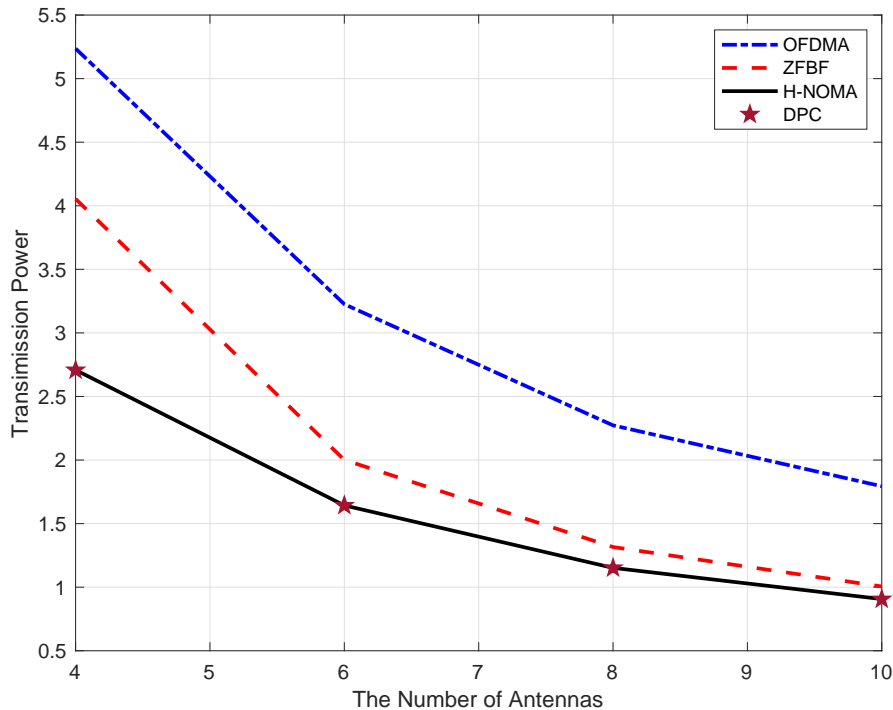}\caption{Transmission power versus the number of antennas.}
\end{figure}

Fig. 4 depicts the BS transmission power versus the number of antennas
respectively using the schemes of hybrid NOMA (H-NOMA), ZFBF, and OFDMA. In Fig. 4,
the IRS is located at the point $\left(20\textrm{m},20\textrm{m}\right)$
and the number of IRS elements is $N=10$. Here, note that all the
schemes are assisted by IRS. The minimum transmission powers achieved
by NOMA and ZFBF use our provided solution. In the IRS-assisted OFDMA
scheme, we first optimize the beamforming with given $\boldsymbol{\Theta}$
and then further optimize $\boldsymbol{\Theta}$ using SDP. As expected,
for the three schemes, the total transmission power decreases with
the increase of the number of antennas. In comparison with ZFBF and
OFDMA schemes, H-NOMA yields a significant performance gain, which is
because NOMA allows the users sharing the same spectrum and spatial
resources and hence improves the performance. Moreover, one can easily
find that the performance achieved by the proposed H-NOMA scheme is
nearly same as DPC, which is owing to the employment of IRS. 

\begin{figure}
\centering
\includegraphics[scale=0.7]{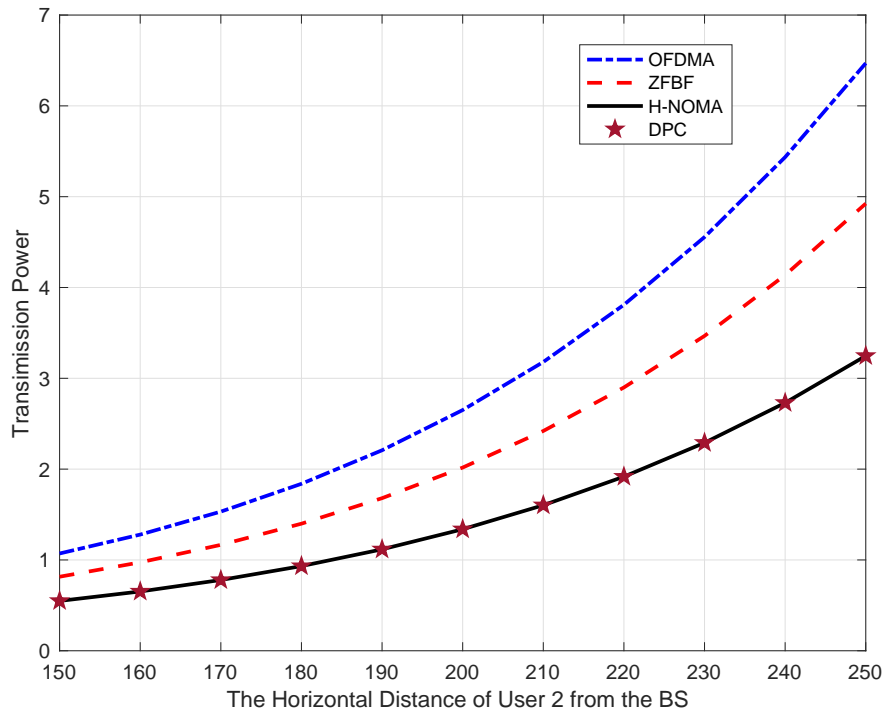}\caption{Transmission power versus the horizontal distance of user 2 from the BS.}
\end{figure}

In Fig. 5, we display the BS transmission power versus the horizontal
distance of user 2 from the BS. Here, the location of user 2 is $\left(D\textrm{m},150\textrm{m}\right)$
and $D$ is the horizontal distance of user 2 from the BS, which ranges
from $150\textrm{m}$ to $250\textrm{m}$ in the simulation. In addition,
in Fig. 5, the locations of user 1 and IRS is respectively $\left(100\textrm{m},100\textrm{m}\right)$
and $\left(20\textrm{m},20\textrm{m}\right)$, the number of antennas
is $M=4$, and the number of IRS elements is $N=10$. As can be seen
in this figure, when user 2 is farther from the BS, the BS needs to
consume more power, which is in line with our expectation. Meanwhile,
some similar phenomenon as Fig. 4 that H-NOMA performs better than
ZFBF and OFDMA and the performance achieved by H-NOMA is nearly same
as DPC. Moreover, one can easily find out that, as the distance
between the two users becomes larger, the H-NOMA scheme yields a significant
performance gain. 

\begin{figure}
\centering
\includegraphics[scale=0.7]{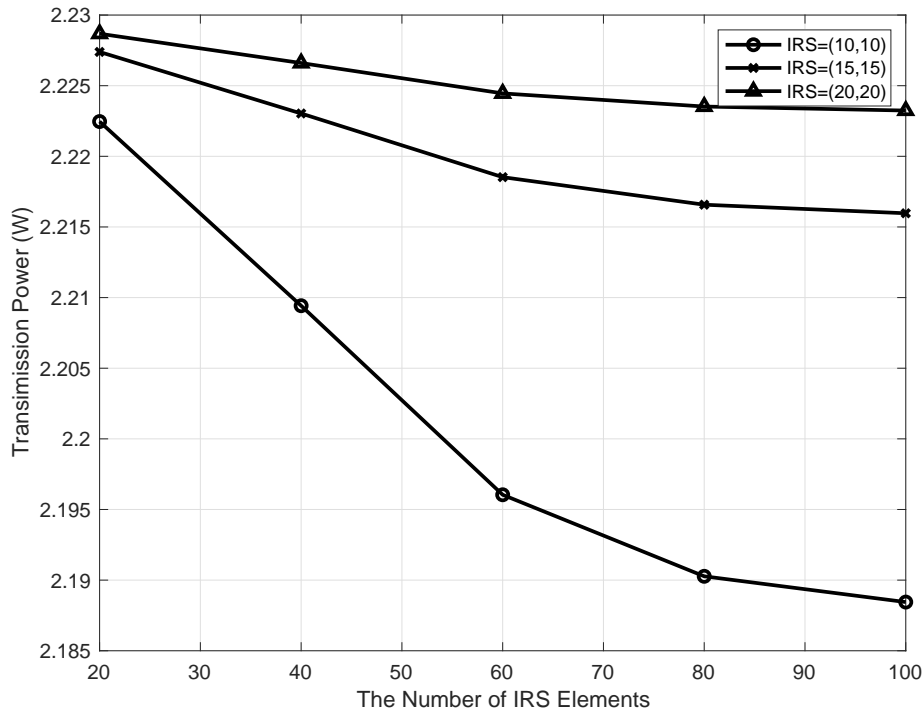}\caption{Transmission power versus the number of IRS elements.}
\end{figure}

Fig. 6 shows the transmission power versus different number of IRS
elements with different IRS locations using the provided  H-NOMA transmission scheme. 
Here, the number of antennas is $M=4$. Obviously, one can
achieve better performance by increasing $N$, which is because a larger
$N$ enables more reflecting elements to receive the signal energy from
the BS. Furthermore, if the IRS is close to the BS, a significant
performance gain can be achieved by increasing $N$. Conversely, there
has no obvious performance gain by increasing $N$ when the user is far
from the BS. Therefore, the number of IRS elements can be properly
selected according to the location of IRS.

\begin{figure}
\centering
\includegraphics[scale=0.7]{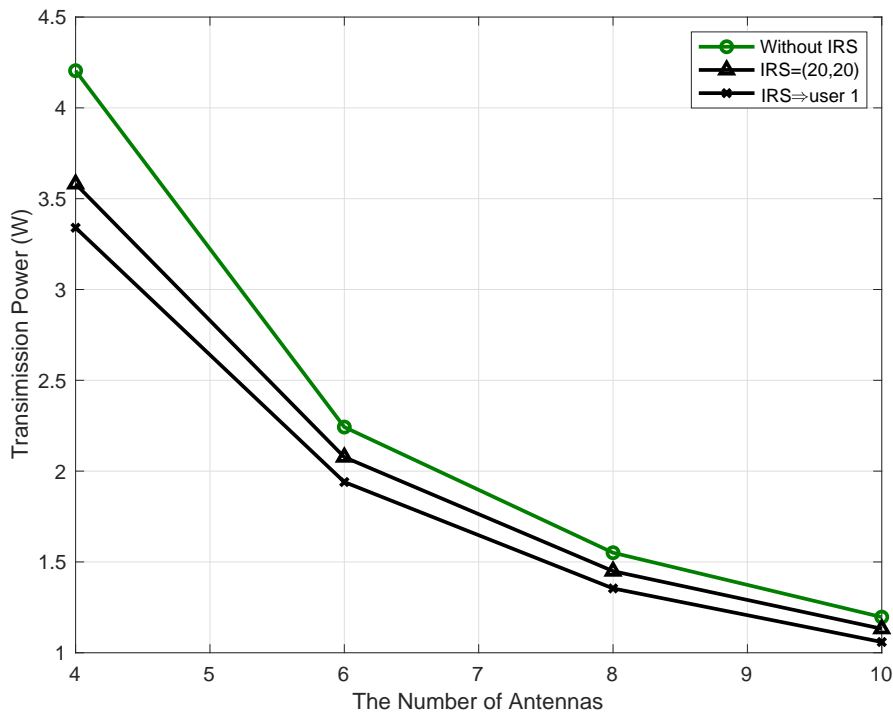}\caption{Transmission power versus the number of antennas.}
\end{figure}

In Fig. 7, the performance both in the IRS-assisted MISO NOMA scheme
and the literature MISO NOMA scheme without IRS is evaluated. The
number of IRS elements is $N=10$. In comparison with the literature
MISO NOMA, the use of IRS in MISO NOMA yields a significant performance
gain, which is because the quasi-degradation condition might be guaranteed
and hence the IRS-assisted MISO NOMA achieves the same performance
as DPC. In addition, in Fig. 7, we also present the performance with
different locations of IRS. It is easy to see that if the location
of IRS is set very close to user 1, the performance is improved. This
also reproves the Corollary \ref{C-q} via simulation. Actually, when
the location of IRS is set very close to user 1, the quasi-degradation
condition can always be satisfied and the performance region can be
obtained in the IRS-assisted MISO NOMA system. 

\section{\label{sec:Conclusion}Conclusion}

In this paper, we optimized the beamforming vectors and the IRS phase shift
matrix respectively for IRS-assisted NOMA and IRS-assisted ZFBF. For IRS-assisted NOMA, we provided the improved quasi-degradation condition under which NOMA can always achieve the same performance as DPC. Furthermore, we show that the improved quasi-degradation condition can be satisfied with greater possibility than the original quasi-degradation condition without IRS. Both for IRS-assisted NOMA and IRS-assisted ZFBF, we characterized the optimal beamforming with given IRS phase shift matrix and proposed algorithms to optimize the phase shift matrix. The IRS-assisted NOMA and the IRS-assisted ZFBF were further compared.  It is shown that, with the same IRS phase shift matrix and improved quasi-degradation condition, NOMA always outperforms ZFBF.  At the same time, we provided certain condition under which the IRS phase shift matrix can be found to generate orthogonal channels and hence the hybrid NOMA transmission scheme was proposed. Simulation results show that the IRS-assisted NOMA system not only outperforms the NOMA system without IRS but also the conventional OMA .

\appendix{}

\subsection{Proof of Proposition 1}

We consider the feasibility of the quasi-degradation constraint in
a converse way. Assume that for each possible $\mathbf{h}_{rk}$, $\mathbf{h}_{dk}$,
$r_{k}^{\min}$, $k=1,2$, and $\boldsymbol{\Theta}$, we always have
\begin{equation}
\frac{1+r_{1}^{\min}}{\cos^{2}\alpha}-\frac{r_{1}^{\min}\cos^{2}\alpha}{\left(1+r_{2}^{\min}\left(1-\cos^{2}\alpha\right)\right)^{2}}>\frac{\left\Vert \boldsymbol{h}_{1}\right\Vert ^{2}}{\left\Vert \boldsymbol{h}_{2}\right\Vert ^{2}}.\label{p-1}
\end{equation}
Since $0\leq\cos^{2}\alpha\leq1$, we should have
\begin{align}
 & \frac{1+r_{1}^{\min}}{\cos^{2}\alpha}-\frac{r_{1}^{\min}\cos^{2}\alpha}{\left(1+r_{2}^{\min}\left(1-\cos^{2}\alpha\right)\right)^{2}}\nonumber \\
\geq & \frac{1+r_{1}^{\min}}{\cos^{2}\alpha}-r_{1}^{\min}\cos^{2}\alpha>\frac{\left\Vert \boldsymbol{h}_{1}\right\Vert ^{2}}{\left\Vert \boldsymbol{h}_{2}\right\Vert ^{2}},
\end{align}
which is equivalently rewritten as 
\begin{equation}
\frac{\left(1+r_{1}^{\min}\right)\left\Vert \boldsymbol{h}_{2}\right\Vert ^{4}}{\boldsymbol{h}_{2}^{H}\boldsymbol{h}_{1}\boldsymbol{h}_{1}^{H}\boldsymbol{h}_{2}}-\frac{r_{1}^{\min}\boldsymbol{h}_{2}^{H}\boldsymbol{h}_{1}\boldsymbol{h}_{1}^{H}\boldsymbol{h}_{2}}{\left\Vert \boldsymbol{h}_{1}\right\Vert ^{4}}>1.\label{p-2}
\end{equation}
 Using \eqref{p-2} and notify $\boldsymbol{h}_{2}^{H}\boldsymbol{h}_{1}\boldsymbol{h}_{1}^{H}\boldsymbol{h}_{2}\leq\boldsymbol{h}_{1}^{H}\boldsymbol{h}_{1}\boldsymbol{h}_{2}^{H}\boldsymbol{h}_{2}$,
we then obtain

\begin{align}
 & \frac{\left(1+r_{1}^{\min}\right)\left\Vert \boldsymbol{h}_{2}\right\Vert ^{4}}{\boldsymbol{h}_{2}^{H}\boldsymbol{h}_{1}\boldsymbol{h}_{1}^{H}\boldsymbol{h}_{2}}-\frac{r_{1}^{\min}\boldsymbol{h}_{2}^{H}\boldsymbol{h}_{1}\boldsymbol{h}_{1}^{H}\boldsymbol{h}_{2}}{\left\Vert \boldsymbol{h}_{1}\right\Vert ^{4}}\nonumber \\
\geq & \frac{\left(1+r_{1}^{\min}\right)\boldsymbol{h}_{2}^{H}\boldsymbol{h}_{2}}{\boldsymbol{h}_{1}^{H}\boldsymbol{h}_{1}}-\frac{r_{1}^{\min}\boldsymbol{h}_{2}^{H}\boldsymbol{h}_{2}}{\boldsymbol{h}_{1}^{H}\boldsymbol{h}_{1}}=\frac{\boldsymbol{h}_{2}^{H}\boldsymbol{h}_{2}}{\boldsymbol{h}_{1}^{H}\boldsymbol{h}_{1}}>1,
\end{align}
implying
\begin{equation}
\boldsymbol{h}_{2}^{H}\boldsymbol{h}_{2}-\boldsymbol{h}_{1}^{H}\boldsymbol{h}_{1}>0.
\end{equation}

Let $\widetilde{\boldsymbol{v}}=\left[\boldsymbol{v};1\right]$ and
$\boldsymbol{v}=[e^{j\theta_{1}},\cdots,e^{j\theta_{N}}]^{H}$, we
can write
\begin{align}
\boldsymbol{h}_{k}^{H}\boldsymbol{h}_{k} & =\left(\mathbf{h}_{rk}^{H}\boldsymbol{\Theta}\mathbf{G}+\mathbf{h}_{dk}^{H}\right)\left(\mathbf{h}_{rk}^{H}\boldsymbol{\Theta}\mathbf{G}+\mathbf{h}_{dk}^{H}\right)^{H},\label{h_k'*h_k}\\
 & =\left(\boldsymbol{v}^{H}\boldsymbol{\Phi}_{k}+\mathbf{h}_{dk}^{H}\right)\left(\boldsymbol{v}^{H}\boldsymbol{\Phi}_{k}+\mathbf{h}_{dk}^{H}\right)^{H}\nonumber \\
 & =\boldsymbol{v}^{H}\boldsymbol{\Phi}_{k}\boldsymbol{\Phi}_{k}^{H}\boldsymbol{v}+\boldsymbol{v}^{H}\boldsymbol{\Phi}_{k}\mathbf{h}_{dk}+\mathbf{h}_{dk}^{H}\boldsymbol{\Phi}_{k}^{H}\boldsymbol{v}+\mathbf{h}_{dk}^{H}\mathbf{h}_{dk},\nonumber \\
 & =\tilde{\boldsymbol{v}}^{H}\boldsymbol{\Upsilon}_{k}\tilde{\boldsymbol{v}}.\nonumber 
\end{align}
Then, for each $\tilde{\boldsymbol{v}}$, we have 
\begin{equation}
\boldsymbol{h}_{2}^{H}\boldsymbol{h}_{2}-\boldsymbol{h}_{1}^{H}\boldsymbol{h}_{1}=\tilde{\boldsymbol{v}}^{H}\left(\boldsymbol{\Upsilon}_{2}-\boldsymbol{\Upsilon}_{1}\right)\tilde{\boldsymbol{v}}>0.
\end{equation}
Therefore,  $\boldsymbol{\Upsilon}_{2}-\boldsymbol{\Upsilon}_{1}\succ\mathbf{0}$.
Hence, one can easily obtain $\lambda_{\min}\left(\boldsymbol{\Upsilon}_{2}-\boldsymbol{\Upsilon}_{1}\right)>0$,
where $\lambda_{\min}\left(\boldsymbol{\Upsilon}_{2}-\boldsymbol{\Upsilon}_{1}\right)$
denotes the minimum eigenvalue of matrix $\boldsymbol{\Upsilon}_{2}-\boldsymbol{\Upsilon}_{1}$. 

Therefore, if $\lambda_{\min}\left(\boldsymbol{\Upsilon}_{2}-\boldsymbol{\Upsilon}_{1}\right)>0,$
the IRS phase shift matrix cannot be found to satisfy the quasi-degradation
constraint. Conversely, if $\lambda_{\min}\left(\boldsymbol{\Upsilon}_{2}-\boldsymbol{\Upsilon}_{1}\right)\leq0$,
i.e., $\lambda_{\max}\left(\boldsymbol{\Upsilon}_{1}-\boldsymbol{\Upsilon}_{2}\right)\geq0$,
the quasi-degradation constraint is always feasible. This completes
the proof. 

\subsection{Proof of Corollary 1}

Assume the IRS is set at the location of user 1, we note that $\left\Vert \mathbf{h}_{r1}\right\Vert ^{2}\rightarrow+\infty$,
and hence $\frac{\left\Vert \boldsymbol{h}_{1}\right\Vert ^{2}}{\left\Vert \boldsymbol{h}_{2}\right\Vert ^{2}}=\frac{\left\Vert \mathbf{h}_{r1}\boldsymbol{\Theta}\mathbf{G}+\mathbf{h}_{d1}\right\Vert ^{2}}{\left\Vert \mathbf{h}_{r2}\boldsymbol{\Theta}\mathbf{G}+\mathbf{h}_{d2}\right\Vert ^{2}}\rightarrow+\infty$.
In addition, the function 
\begin{equation}
L\left(\cos^{2}\alpha\right)=\frac{1+r_{1}^{\min}}{\cos^{2}\alpha}-\frac{r_{1}^{\min}\cos^{2}\alpha}{\left(1+r_{2}^{\min}\left(1-\cos^{2}\alpha\right)\right)^{2}}
\end{equation}
is monotonically decreasing function of $\cos^{2}\alpha$. The supposition
$\boldsymbol{h}_{1}^{H}\boldsymbol{h}_{2}\neq0$ implies that $\cos^{2}\alpha\geq\kappa>0$.
Therefore,  $L\left(\cos^{2}\alpha\right)\leq L\left(\kappa\right)<+\infty$,
this in turn implies that the quasi-degradation condition, \eqref{quasi-degradation}, 
can always be satisfied. 

\subsection{Proof of Proposition 2}

Since $\boldsymbol{h}_{k}\boldsymbol{h}_{k}^{H}=\boldsymbol{\chi}_{k}\boldsymbol{Q}\boldsymbol{\chi}_{k}^{H}$
for $k=1,2$, using condition \eqref{quasi-transfer}, we always have
\begin{equation}
\boldsymbol{h}_{1}\boldsymbol{h}_{1}^{H}-\left(1+r_{1}^{\min}\right)\boldsymbol{h}_{2}\boldsymbol{h}_{2}^{H}\succeq\boldsymbol{0},
\end{equation}
and hence
\begin{equation}
\boldsymbol{h}_{2}^{H}\left(\boldsymbol{h}_{1}\boldsymbol{h}_{1}^{H}-\left(1+r_{1}^{\min}\right)\boldsymbol{h}_{2}\boldsymbol{h}_{2}^{H}\right)\boldsymbol{h}_{2}\geq0,\label{proof-p4-1}
\end{equation}
which is equivalently written as 
\begin{equation}
\frac{1+r_{1}^{\min}}{\cos^{2}\alpha}\leq\frac{\left\Vert \boldsymbol{h}_{1}\right\Vert ^{2}}{\left\Vert \boldsymbol{h}_{2}\right\Vert ^{2}}.\label{1}
\end{equation}
In addition, according to the quasi-degradation condition \eqref{quasi-degradation},
with $0\leq\cos^{2}\alpha\leq1$, we have
\begin{equation}
\frac{1+r_{1}^{\min}}{\cos^{2}\alpha}-\frac{r_{1}^{\min}\cos^{2}\alpha}{\left(1+r_{2}^{\min}\left(1-\cos^{2}\alpha\right)\right)^{2}}\leq\frac{1+r_{1}^{\min}}{\cos^{2}\alpha}.\label{2}
\end{equation}
Combining \eqref{1} and \eqref{2}, we have 
\begin{equation}
\frac{1+r_{1}^{\min}}{\cos^{2}\alpha}-\frac{r_{1}^{\min}\cos^{2}\alpha}{\left(1+r_{2}^{\min}\left(1-\cos^{2}\alpha\right)\right)^{2}}\leq\frac{\left\Vert \boldsymbol{h}_{1}\right\Vert ^{2}}{\left\Vert \boldsymbol{h}_{2}\right\Vert ^{2}},
\end{equation}
which is in fact the quasi-degradation constraint shown in \eqref{quasi_degradation_P7-1-1}. 

\subsection{Proof of Proposition 3}

Since $0\leq\cos^{2}\alpha\leq1$, we have 
\begin{align}
F & =\phi_{1}^{2}\left(\left(1+r_{2}^{\min}\right)^{2}-\left(2+r_{2}^{\min}\right)r_{2}^{\min}\cos^{2}\alpha\right)+\phi_{2}^{2},\nonumber \\
 & =\frac{\sigma^{2}\left(r_{1}^{\min}\left(1+r_{2}^{\min}\right)\right)}{\textrm{ Tr}\left(\boldsymbol{Q}\boldsymbol{\Upsilon}_{1}\right)}\frac{1}{1+r_{2}^{\min}\left(1-\cos^{2}\alpha\right)}+\frac{\sigma^{2}r_{2}^{\min}}{\textrm{ Tr}\left(\boldsymbol{Q}\boldsymbol{\Upsilon}_{2}\right)},\nonumber \\
 & \leq\frac{\sigma^{2}\left(r_{1}^{\min}\left(1+r_{2}^{\min}\right)\right)}{\textrm{ Tr}\left(\boldsymbol{Q}\boldsymbol{\Upsilon}_{1}\right)}+\frac{\sigma^{2}r_{2}^{\min}}{\textrm{ Tr}\left(\boldsymbol{Q}\boldsymbol{\Upsilon}_{2}\right)},
\end{align}
which completes the proof.

\subsection{Proof of Lemma 3}

By using the least square property of Moore-Penrose inverse \cite{horn2012matrix},
the optimal solution to problem \eqref{P1-ZF} with given $\boldsymbol{\Theta}$
is trivially obtained. Specifically, let $\boldsymbol{H}=\left[\boldsymbol{h}_{1},\boldsymbol{h}_{2}\right]$,
the optimal solution is given as
\begin{align}
\left[\boldsymbol{w}_{1}^{*},\boldsymbol{w}_{2}^{*}\right] & =\boldsymbol{H}^{H\dagger}\left[\begin{array}{cc}
\sqrt{r_{1}^{\min}} & 0\\
0 & \sqrt{r_{2}^{\min}}
\end{array}\right]\nonumber \\
 & =\frac{\left[\sqrt{r_{1}^{\min}}\boldsymbol{a}_{1},\sqrt{r_{2}^{\min}}\boldsymbol{a}_{2}\right]}{\left\Vert \boldsymbol{h}_{1}\right\Vert ^{2}\left\Vert \boldsymbol{h}_{2}\right\Vert ^{2}\sin^{2}\alpha},
\end{align}
where $\dagger$ represents the Moore-Penrose inverse. 

\subsection{Proof of Theorem 1}

Note that 
\begin{equation}
\frac{r_{2}^{\min}\sigma^{2}}{\left\Vert \boldsymbol{h}_{2}\right\Vert ^{2}}\leq\frac{\sigma^{2}}{\sin^{2}\alpha}\frac{r_{2}^{\min}\sigma^{2}}{\left\Vert \boldsymbol{h}_{2}\right\Vert ^{2}},
\end{equation}
and 
\begin{align}
\frac{r_{1}^{\min}\left(1+r_{2}^{\min}\right)\sigma^{2}}{\left\Vert \boldsymbol{h}_{1}\right\Vert ^{2}\left(1+r_{2}^{\min}\sin^{2}\alpha\right)} & \leq\frac{r_{1}^{\min}\left(1+r_{2}^{\min}\right)\sigma^{2}}{\left\Vert \boldsymbol{h}_{1}\right\Vert ^{2}\left(\sin^{2}\alpha+r_{2}^{\min}\sin^{2}\alpha\right)}\nonumber \\
 & \leq\frac{r_{1}^{\min}\sigma^{2}}{\left\Vert \boldsymbol{h}_{1}\right\Vert ^{2}\sin^{2}\alpha}.
\end{align}
Therefore, 
\begin{align}
P^{NOMA} & =\frac{r_{1}^{\min}\left(1+r_{2}^{\min}\right)\sigma^{2}}{\left\Vert \boldsymbol{h}_{1}\right\Vert ^{2}\left(1+r_{2}^{\min}\sin^{2}\alpha\right)}+\frac{r_{2}^{\min}\sigma^{2}}{\left\Vert \boldsymbol{h}_{2}\right\Vert ^{2}}\nonumber \\
 & \leq\frac{r_{1}^{\min}\sigma^{2}}{\left\Vert \boldsymbol{h}_{1}\right\Vert ^{2}\sin^{2}\alpha}+\frac{\sigma^{2}}{\sin^{2}\alpha}\frac{r_{2}^{\min}\sigma^{2}}{\left\Vert \boldsymbol{h}_{2}\right\Vert ^{2}}\nonumber \\
 & =\frac{\sigma^{2}}{\sin^{2}\alpha}\left(\frac{r_{1}^{\min}}{\left\Vert \boldsymbol{h}_{1}\right\Vert ^{2}}+\frac{r_{2}^{\min}}{\left\Vert \boldsymbol{h}_{2}\right\Vert ^{2}}\right)=P^{ZFBF}.
\end{align}

\subsection{Proof of Proposition 4}

Given $\boldsymbol{h}_{1}^{H}\boldsymbol{h}_{2}=0$, we have $\cos^{2}\alpha=0$
and hence
\[
\frac{1+r_{1}^{\min}}{\cos^{2}\alpha}-\frac{r_{1}^{\min}\cos^{2}\alpha}{\left(1+r_{2}^{\min}\left(1-\cos^{2}\alpha\right)\right)^{2}}\rightarrow+\infty,
\]
which implying the quasi-degradation condition in \eqref{quasi-degradation}
is invalid. 

Here, we prove that $\boldsymbol{h}_{1}^{H}\boldsymbol{h}_{2}=0$
can be satisfied if and only if $\boldsymbol{R}=\boldsymbol{0}$.
First, we prove the necessary condition. If $\boldsymbol{h}_{1}^{H}\boldsymbol{h}_{2}=0$,
i.e., the channel vectors of the two users are orthogonal, we have
\begin{equation}
\boldsymbol{h}_{1}^{H}\boldsymbol{h}_{2}=\widetilde{\boldsymbol{v}}^{H}\boldsymbol{R}\widetilde{\boldsymbol{v}}=0.\label{h1'*h2-1}
\end{equation}
Since $\widetilde{\boldsymbol{v}}\neq\boldsymbol{0}$, \eqref{Q=00003D0}
implies that $\boldsymbol{R}$ is a zero matrix, i.e., condition \eqref{Q=00003D0}.

We then prove the sufficiency. Given $\boldsymbol{R}=\boldsymbol{0}$,
we have $\boldsymbol{v}^{H}\boldsymbol{R}\boldsymbol{v}=0$, for any
possible $\boldsymbol{v}$. In addition, from \eqref{h1'*h2-1}, one
can easily find out that 
\begin{equation}
\boldsymbol{v}^{H}\boldsymbol{R}\boldsymbol{v}=\boldsymbol{h}_{1}^{H}\boldsymbol{h}_{2}=0.
\end{equation}

From Lemma \ref{L-ZF-B}, if $\boldsymbol{h}_{1}^{H}\boldsymbol{h}_{2}=0$,
the minimum transmission power of the BS achieved by using ZFBF is
\begin{equation}
p=\sigma^{2}\left(\frac{r_{1}^{\min}}{\left\Vert \boldsymbol{h}_{1}\right\Vert ^{2}}+\frac{r_{2}^{\min}}{\left\Vert \boldsymbol{h}_{2}\right\Vert ^{2}}\right),
\end{equation}
 which also satisfies
\begin{equation}
p\leq P^{ZFBF}=\frac{\sigma^{2}}{\sin^{2}\alpha}\left(\frac{r_{1}^{\min}}{\left\Vert \boldsymbol{h}_{1}\right\Vert ^{2}}+\frac{r_{2}^{\min}}{\left\Vert \boldsymbol{h}_{2}\right\Vert ^{2}}\right),\label{G<p_zfbf}
\end{equation}
and 
\begin{equation}
p\leq P^{NOMA}=\frac{r_{1}^{\min}\left(1+r_{2}^{\min}\right)\sigma^{2}}{\left\Vert \boldsymbol{h}_{1}\right\Vert ^{2}\left(1+r_{2}^{\min}\sin^{2}\alpha\right)}+\frac{r_{2}^{\min}\sigma^{2}}{\left\Vert \boldsymbol{h}_{2}\right\Vert ^{2}}.\label{g<p_noma}
\end{equation}
From \eqref{G<p_zfbf} and \eqref{g<p_noma}, one can easily find
that if the channels of user 1 and user 2 are orthogonal, the ZFBF
scheme obtains the optimal performance. 

\bibliographystyle{IEEEtran}
\bibliography{MISO_NOMA_LIS}

\end{document}